\documentclass[12pt]{article}
\usepackage{amsmath}
\usepackage{graphicx,psfrag,epsf,ulem,bm}
\usepackage{enumerate}
\usepackage{hyperref}
\usepackage{url} 
\usepackage{amssymb}   
\usepackage{xcolor}    
\usepackage{amsthm}    
\usepackage{csquotes}  
\usepackage{subcaption}  
\usepackage{floatrow}
\usepackage{setspace}
\usepackage{adjustbox}
\usepackage[top=1in,bottom=1in,left= 1in, right= 1in]{geometry} 
\usepackage{chngcntr}
\usepackage{apptools}

\newcommand{\blind}{1}

\makeatletter
\newcommand*{\rom}[1]{\expandafter\@slowromancap\romannumeral #1@}  
\makeatother

\def \ben{\begin{eqnarray*}}
	\def \een{\end{eqnarray*}}
\def \bea{\begin{eqnarray}}
\def \eea{\end{eqnarray}}

\theoremstyle{definition}

\def\pf{{\bf Proof. }}
\newtheorem{theorem}{Theorem}[section]
\newtheorem{lemma}{Lemma}[section]

\usepackage{soul, multirow}
\newcommand{\DNA}[1]{{\color{black} #1}}

\begin{document}

	\def\spacingset#1{\renewcommand{\baselinestretch}%
		{#1}\small\normalsize} \spacingset{1}

	
	\if1\blind
	{
		\title{\bf Two-Sample High Dimensional Mean Test Based On Prepivots}
		\author{Santu Ghosh\thanks{Corresponding Author} \\
Augusta University \\
sghosh@augusta.edu\\[10pt]
			Deepak Nag Ayyala \\
			Augusta University \\
			dayyala@augusta.edu\\[10pt]
			Rafael Hellebuyck\\ 
			South Carolina Department of Public Safety
}
		\clearpage\maketitle
\thispagestyle{empty}
	} \fi
	
	\if0\blind
	{
		\bigskip
		\bigskip
		\bigskip
		\begin{center}
			{\LARGE\bf High Dimensional Mean Test: Random Projection and Prepivotings}
		\end{center}
		\medskip
	} \fi
	


\begin{abstract}
\DNA{Testing equality of mean vectors is a very commonly used criterion when comparing two multivariate random variables}.  \DNA{Traditional} tests such as Hotelling's $T^{2}$ become either unusable or output small power when the number of variables is greater than the combined sample size. \DNA{In this paper, w}e propose a test using both prepivoting and Edgeworth expansion  for testing the equality of two population mean vectors in the ``large $p$\DNA{,} small n" setting. The asymptotic null distribution of the test statistic is derived and it is shown that the power of suggested test converges to one under certain alternatives when both $n$ and $p$ increase to infinity  
against sparse alternatives. Finite sample performance of the proposed test statistic is compared with other recently developed tests designed to also handle the ``large $p$\DNA{,} small $n$" situation through simulations. The proposed test achieves competitive rates for both type I error \DNA{rate} and power. The usefulness of our test is illustrated by applications to two microarray gene expression data sets.\newline

\noindent
\textit{keywords}:  High dimensional mean vectors test; Prepivoting;  Edgeworth expansion; Limiting null distribution; Power; Dense alternatives; Sparse alternatives

\maketitle
\end{abstract}
\newpage
\spacingset{1.75}
\section{Introduction}
Due to the advancement of technologies, researchers in various fields of sciences often encounter the analysis of massive data sets to extract useful information for new scientific discoveries. Many of these data sets contain a large number of features\DNA{, often exceeding the}  number of observations. Such data sets  are common in the fields of microarray gene expression data analysis \cite{1,2}; medical image analysis \cite{3}, signal processing, astrometry and finance \cite{4}. \DNA{Analysis of s}uch high dimension, low sample size data \DNA{sets} present a substantial challenge, known as the ``large $p$\DNA{,} small $n$" problem. \DNA{One of the most prominent problems that} researchers are interested in \DNA{is} making inferences on the mean structure of \DNA{the} population. However, many well known classical multivariate methods cannot be used to make inferences on the mean structures for large $p$ small $n$  data. For example, because of the singularity of the estimated pooled covariance matrix, the classical Hotelling’s $T^2$ statistic \cite{5} breaks down for the two sample test when the dimension of the data exceeds the combined sample size. 

Over the last few years, researchers have turned their attention to developing statistical methods that can handle the large $p$ small $n$ problem. A series of important works 
have been done on the two-sample location problem in the large $p$ small $n$  setup, and several parametric and non-parametric tests have been developed.
Most of these prior works \DNA{are primarily motivated by} creating methodologies that avoid the issues that Hotelling's $T^{2}$\DNA{-statistic} faces in the large $p$ small $n$ scenario, particularly the \DNA{singularity} of the pooled sample covariance matrix. An early important work in this direction was developed by Bai and Saranadasa \cite{6} where they considered the squared Euclidean norm between sample means of the two populations as an alternative to the Mahalanobis distance used in Hotelling’s $T^{2}$-statistic. However, one of the criticisms (see  \cite{7}) of this method is the assumption of equal variance structures for \DNA{the two populations}, which is hard to verif\DNA{y} for high-dimensional data. \DNA{Addressing this shortcoming}, Chen and Qin \cite{7} suggest\DNA{ed} an alternative method that does not assume the equality of variance-covariance matrices and removes the cross-product terms from the square\DNA{d} Euclidean norm \DNA{of difference of sample means between the two populations}. Another method proposed by Srivastava \cite{8} replaces the inverse of the sample covariance matrix  in Hotelling’s $T^{2}$-statistic with the inverse of the diagonal of the covariance matrix. \DNA{Modifications to this work \cite{9,10} have relaxed some of the assumptions, resulting in relatively similar test statistics with improved performance.} Recently, Gregory et al. \cite{11} suggest\DNA{ed} a test procedure, known as generalized component test (GCT), that bypasses the full estimation of the covariance matrix by assuming that  there is a logical ordering among \DNA{the} $p$ components in such a way that the dependence between any two components is related to their displacement. In a development of another direction, Cai et al. \cite{4} suggest a test procedure based on linear transformation of the data by a precision matrix, where the test statistic is the sup-norm of marginal $t$-statistics of transformed data. Tests developed in \cite{6,7,8,9,11} are designed for testing dense but possibly weak signals, i.e., there are a large number  of small non-zero means differences. Test suggested in \cite{4} is designed for testing  sparse signals, i.e., when there are a small number  of large non-zero means differences. The implementation of test \cite{4}   requires  sparsity structure assumptions on the covariance or precision matrix, which may not be satisfied in real applications. Furthermore, it also requires an estimate of the $p\times p$ precision matrix, which is time-consuming for large $p$, see \cite{12}.

Driven by the above two concerns of the test \cite{4}, we reconsider the problem of testing two mean vectors under sparse alternatives. We propose a new test based on the prepivoting approach. The concept of  prepivoting is introduced by  Beran \cite{12}. Prepivoting is a mechanism  that transforms a root, a function of samples and parameters, by its estimated distribution function to a random variable, known as prepivoted root. The characteristics of the distribution of a prepivoted root are very similar to that of a  uniform distribution on (0,1). Thus, the distribution of the prepivoted root is less dependent on population parameters than  the distribution of the root. Consequently, approximate inference based on the prepivoted root is more accurate than approximate inference based on a original root. Pre-pivoting has not been considered previously in high-dimensional testing problems. A snapshot of the development of our proposed test method is as follows: we construct prepivoting  marginally for all the $p$ variables individually using an asymptotic refinement, particularly the Egdeworth exapnsion. Then the marginal prepivoted roots are combined using the max-norm. 

The rest of paper is organized as follow. In Section \ref{sec:methodology}, we formulate the construction of the  suggested test method for the hypothesis testing problem (\ref{eq1})  based on the prepivoting technique technique and study the limiting distribution null distribution of the test statistic. In this section, we also analyze the power of the suggested test. Simulation studies are presented in Section \ref{sec:simulation} and Section \ref{sec:realdata} presents application of the proposed test to two microarray gene expression data sets. We finally conclude the article with some brief remarks in Section 5. Proofs and other theoretical details are relegated to the Appendix.

\section {Methodology}
\label{sec:methodology}

Let $\{\bm{X}_{1}, \bm{X}_{2},\ldots,\bm{X}_{n}\}$ and $\{\bm{Y}_{1}, \bm{Y}_{2},\ldots,\bm{Y}_{m}\}$ \DNA{be independently and identically distributed ({\it i.i.d.})} samples drawn from two $p$ variate distributions $F_{x}$ and $F_{y}$, respectively. Let $\bm{\mu}_{x}$ and $\Sigma_x$  denote the mean vector and covariance matrix of $F_1$ respectively; $\bm{\mu}_{y}$ and $\Sigma_y$  denote the mean vector and covariance matrix of $F_2$ respectively. The primary focus is on testing equality of mean vectors of the two populations,
\begin{equation} \label{eq1}\text{H}_0: \bm{\mu}_{x}=\bm{\mu}_{y}\;\;\text{vs.}\;\; \text{H}_1: \bm{\mu}_{x}\neq \bm{\mu}_{y}.\end{equation} 
The \DNA{proposed} method for the two sample mean vector test (\ref{eq1}) adopts the prepivoting technique introduced by Beran \cite{12}. This technique is useful to draw more accurate statistical inferences in the absence of a pivotal quantity. Prepivoting technique transforms the original root, can be functionally dependent on data and parameters,  into another value through its estimated cumulative distribution function (cdf). The new value is  known as prepivoted root or prepivot.  

\subsection{Overview On Prepivoting}
For the sake of simplicity, \DNA{we firstly introduce the concept of prepivoting}. Consider a statistical model $F_{\theta}$, where $\theta\in\Theta\subset \mathbb{R}$ is unknown and \DNA{we are interested in }the hypothesis test, $H_{0}:\theta=\theta_{0}$ vs.  $H_1:\theta\neq\theta_{0}$.  Let $\bm{U}$ be a random sample from $F\DNA{_{\theta}}$.  In frequentist inference problem, one often compares the observed \DNA{value} of a root, $R(\bm{U},\theta)$, under $H_{0}$ to the quantile of that root's null distribution. We reject $H_{0}$ at a significance level $\alpha\in(0,1)$ whenever $R(\bm{U},\theta_{0})$  is greater than $J^{-1}(1-\alpha)$,  $J(v)=P\{R(\bm{U},\theta_{0})\leq v\}$ is cdf of the root $R(\bm{U},\theta)$ under $H_{0}$. 
Unless $R(\bm{U}, \theta)$ is a pivotal quantity, $J^{-1}(1-\alpha)$ depends on $F$. \DNA{If $R(\bm{U}, \theta)$ is not a pivotal quantity, then its null distribution can be derived only through $F_{\theta}$}. Thus a hypothesis test based on a non pivotal quantity is not exact. \DNA{Furthermore}, the actual level of a test based on a non pivotal root differs from the nominal level. 

In such cases, the prepivoting is \DNA{a} useful \DNA{technique} to reduce errors associated with inference. The main idea behind Beran's prepivoting technique is to transform  the original root $R(U,\theta)$ to a new root $\hat{J}(R(U,\theta))$ whose sampling distribution  is ``smoother" in the sense that its dependence on the  population distribution is reduced, where $\hat{J}(\cdot)$ is  an estimate of $J(\cdot)$. 


If the distribution $F$ is unknown, then prepivots are constructed based on the empirical distribution function $\hat{F}$. Let $\hat{J}(x)=P\{T(\bm{U^{\star}},\hat{\theta})<x|\hat{F}\}$ denote the bootstrap estimate of $J(x)$, and $\bm{U^{\star}}$ \DNA{be} a random sample from $\hat{F}$ . The prepivot is given by $\hat{J}(R(\bm{U}, \theta))$, and \DNA{it has been established} \cite{12} \DNA{that} $\hat{J}(R(\bm{U}, \theta))$ is closer \DNA{to} being pivotal than the original root or more precisely an approximately U(0,1) random variable.  Let $\hat{J}_{1}$ denote the bootstrap estimate of $J_{1}(x)$, the cdf of $\hat{J}( R(\bm{U}, \theta) )$. The statistical test based on prepivoting rejects $H_{0}:\theta=\theta_{0}$ at level $\alpha$  if $R(\bm{U},\theta_{0})>\hat{J}^{-1}(\hat{J}_{1}(1-\alpha))$.  Prepivoting usually is accomplished by a Monte Carlo simulation (\cite{12,13}) known as nested double-bootstrap
algorithm. This algorithm  consists of an outer and an inner level of non\DNA{-}parametric bootstrap, which give the estimated cdfs $\hat{J}(x)$ and $\hat{J}_{1}(x)$, respectively.  However, prepivoting  has been criticized for its computational cost due to the nested double-bootstrap
algorithm (\cite{14}; \cite{15}).


\subsection{Proposed Test Method}
To facilitate the construction of \DNA{the proposed test}, consider the elements of the random vectors, $\mathcal{X}_i = \{{X}_{1i}, \ldots, {X}_{ni}\}$ and $\mathcal{Y}_i = \{{Y}_{1i}, \ldots, {Y}_{mi}\}$, $i = 1, \ldots, p$. As a procedure for the two-sample mean vector test given in \eqref{eq1} based on prepivoting,  consider the root
\begin{equation}
    R_{i}=R_{i}(\mathcal{X}_{i}, \mathcal{Y}_{i},\mu_{xi},\mu_{yi})=  \frac{|(\bar{X}_{i}-\bar{Y}_i- (\mu_{xi}-\mu_{yi})|}{\sqrt{\hat{\sigma}^2_{xi}/n +\hat{\sigma}^2_{yi}/m}},\;\; i=1,\ldots,p;
\end{equation}
where $\bar{X}_{i}$ -$\bar{Y}_{i}$ is the $i^{\rm th}$ component of $\bm{\bar{X}}-\bm{\bar{Y}}$, which is an unbiased estimator for the $i^{\rm th}$ element of $\bm{\mu}_{x}-\bm{\mu_{y}}$, {\it viz.}  $\mu_{xi}-\mu_{yi}$. The quantities $\hat{\sigma}^2_{xi}=n^{-1}\sum_{j=1}^{n}(X_{ij}-\bar{X}_{i})^2$ and $\hat{\sigma}^2_{yi}=m^{-1}\sum_{j=1}^{m}(Y_{ij}-\bar{Y}_{i})^2$ are plug\DNA{-}in estimators of $i^{\rm th}$ diagonal element of $\Sigma_x$ and $\Sigma_y$, respectively. Let $R^{0}_{i}$ denote the version of $R_{i}$ under $H_{0}$. Then $R^{0}_{i}=\frac{|\bar{X}_{i}-\bar{Y}_i|}{\sqrt{\hat{\sigma}^2_{xi}/n_{1}+\hat{\sigma}^2_{yi}/n_{2}}}$, $i=1\ldots, p$. 

Here, we are interested for detecting relatively sparse signals so statistics  of the maximum-type such are more useful than sum of squares-type statistics (see \cite{4}). Thus, at first glance, the test statistic  $\max_{1\leq i\leq p} R^{0}_{i}$  is a potential candidate for the testing problem \ref{eq1}. 
\DNA{Under normality, the roots follow $t$ distribution with $n_1 + n_2 - 2$ degrees of freedom under $H_0$, provided $\sigma_{xi} = \sigma_{yi}$. For unequal variances, we encounter the Behrens-Fisher problem and  even under normality, the roots are  non-pivotal. For any general distribution, deriving  a pivotal root is highly infeasible  and hence the corresponding inferences are approximate inferences. To overcome the drawback of non-pivotal roots, the roots can be pre-pivoted using its  empirical distribution } for reducing errors involved in approximate inferences.

Suppose that $J_{i}(\cdot)$ denotes the cdf of $R^{0}_{i}$, whic is generally unknown. Let $\hat{J}_{i}(\cdot)$ denote the bootstrap estimator of $J_{i}(\cdot)$. Instead of $R^{0}_{i}$'s, we consider the prepivots or the transformed roots $\hat{J}_{i}(R_{i})\; i=1,\ldots,p$  to construct a test for two-sample mean vectors test. \DNA{A}s discussed above in Subsection 2.1,  $\hat{J}_{i}(R^{0}_{i})$ is \DNA{a} better pivot than \DNA{the} original root $R^{0}_{i}$. 

Our test procedure is based on the intuition that if \DNA{only} a \DNA{very} few components of $\bm{\mu}_x$ and $\bm{\mu}_y$ are unequal, then a test based on the max-norm should detect that $\bm{\mu}_x$ and $\bm{\mu}_y$ are different with greater power than tests based on sum-of-squares of sample mean differences. With this intuition, a natural choice is the test based on the maximum norm of $\hat{J}_{1}(R^{0}_{1}),\ldots, \hat{J}_{p}(R^{0}_{p})$. Thus, we define the test statistic for the two-sample mean vector test as
\begin{equation} 
T=\max_{1\leq i \leq p}\hat{J}_{i}(R^{0}_{i}).
\end{equation}
The main advantage of $T$ over $\max_{1\leq i \leq p}R^{0}_{i}$ is that  $T$ is less dependent on the characteristics of the population distributions $F_x$ and $F_y$.  Due to the unavailability of the sampling distribution of $T$,  we can approximate the critical value $J^{-1}(1-\alpha)$ or the p-value $P\{T>t\}$ of a test based on $T$ by means of nonparamteric bootstrap, where $J$ is the CDF of $T$.  This approach can be  challenging because the computational complexity grow\DNA{s with}  the dimension $p$ at a linear rate. To avoid such computational burden, we replace each $\hat{J}_{i}(R^{0}_{i})$ by an analytical approximation based on Edgeworth expansion theory.  


\subsection{Edgeworth expansion based prepivots and the suggested test statistic}
The Edgeworth expansion is a sharpening of the central limit approximation to include \DNA{higher-order moment} terms involving the  skewness and kurtosis of the population distribution. The development of the Edgeworth expansion begins with the derivation of series expansion for the density and distribution
functions of normalized \DNA{\it i.i.d.} sample means \cite{16}. A rigours treatment of these expansion was considered by Cram\'{e}r's \cite{17}. Later on, Bhattacharya and Ghosh (\cite{18}) developed the Edgeworth expansion theory for a smooth function of \DNA{\it i.i.d.} sample mean vectors.  

To construct Edgeworth expansion based prepivots, it is necessary to obtain the Edgeworth expansion of $J_{i}(\cdot)$, the CDF of $R^{0}_i$. Theorem \ref{pro1} provides the  such asymptotic expansions. Throughout this paper, we assume that 
\begin{itemize}
\item[(A1)] $\sup_{i\leq p} \mathbb{E}(X^{8}_{i}) <\infty$, and $\sup_{i\leq
p} \mathbb{E}(Y^{8}_{i}) <\infty$. That is, marginal eighth order moments of the random vectors $\bm{X}=(X_1,\ldots, X_{p})^{\top}$ and $\bm{Y}=(Y_{1},\ldots, Y_{p})^{\top}$ are uniformly bounded.

\item [(A2)] For each $i\in\{1,\ldots, p\}$, the distributions of $(X_{i}, X^{2}_{i})$ and $(Y_{i}, Y^{2}_{i})$ satisfy bivariate Cram\'{e}r's continuity condition (\cite{18}; pp.66-67 in \cite{19}).

\item[(A3)] $n/N \to c$   as $N=n+m\to\infty$. 
\end{itemize} 

\begin{theorem}\label {pro1}
Under the assumptions (A1), (A2), and (A3), the CDF of $R^{0}_{i}$ has the following Edgeworth expansion
\begin{equation}  
P(R_{i}\leq x)=2\Phi(x)-1+2N^{-1}q_{i}(x)\phi(x)+o(N^{-1}),
\end{equation}
holds uniformly in $i$ any $x\geq0$ ; where
\begin{gather*}
q_{i}(x)=x\left[\frac{1}{12}\eta_{1,i}^{-2}\eta_{3,i}(x^2-3)-\frac{1}{18}\eta^{-3}_{1,i}\eta^{2}_{2,i}(x^4+2x^2-3)-\frac{1}{4}\eta^{-2}_{1,i}\{\eta_{4,i} (x^{2}+3)+2\frac{\sigma^{2}_{x,i} \sigma^{2}_{y,i}}{r^{2}_{x}r^{2}_{y}}\}\right],\\
\eta_{1,i}=\frac{\sigma^{2}_{x,i}}{r_{x}}+\frac{\sigma^{2}_{y,i}}{r_{y}}\;\mbox{,}\;\eta_{2,i}=\frac{\gamma_{x,i}}{r^{2}_{x}}-\frac{\gamma_{y,i}}{r^{2}_{y}}\;\mbox{,}\; \eta_{3,i}=\frac{\kappa_{x,i}}{r^{3}_{x}}+\frac{\kappa_{y,i}}{r^{3}_{y}}\;\;\mbox{and}\;\; \eta_{4,i}=\frac{\sigma^{4}_{x,i}}{r^{3}_{x,i}}+\frac{\sigma^{4}_{y,i}}{r^{3}_{y,i}}.
\end{gather*}
\end{theorem}
\noindent \pf See Appendix

In Theorem \ref{pro1}, $\gamma_{x,i}$ and $\gamma_{y,i}$ denote the skewness of the $i^{\rm th}$ component of $\bm{X}$ and $\bm{Y}$, respectively and $\kappa_{x,i}$ and $\kappa_{y,i}$ denote the kurtosis of the $i^{\rm th}$ component of $\bm{X}$ and $\bm{Y}$, respectively;  where 
\ben
\gamma_{x,i}&=&E{(X_{i}-\mu_{x,i})^{3}}\;\mbox{,}\; \gamma_{y,i}=E{(Y_{i}-\mu_{y,i})^{3}},\\ \kappa_{x,i}&=&E{(X_{i}-\mu_{x,i})^{4}}-3\sigma^{4}_{x,i}\;\mbox{and}\; \kappa_{y,i}=E{(Y_{i}-\mu_{y,i})^{4}}-3\sigma^{4}_{y,i}.
\een
The bootstrap estimators $\hat{J}_{i}(\cdot)$ have similar asymptotic expansions as $J_{i}(\cdot)$ given in Theorem \ref{pro1} but in their expansion the population moments are replaced with the corresponding sample moments.  Sample analogues of asymptotic expansions in Proposition \ref{pro1} are
\bea \label{eq2}
\hat{J}_i(x)&=&2\Phi(x)-1+2N^{-1}\hat{q}_{i}(x)\phi(x)+o_{p}(N^{-1}),
\eea
for any $x\geq 0$, where  $\hat{q}_{i}(x)$ are the sample versions of $q_{i}(x)$. The asymptotic expansions in (\ref{eq2}) suggest the following \DNA{second-order} analytical approximations to prepivoted test statistics $\hat{J}_{i}(R^{0}_{i})$:
\bea \label{eq3}
\tilde{J}_{i}(R^{0}_{i})=2\Phi(R^{0}_{i})-1+ 2N^{-1}\hat{q}_{i}(R^{0}_{i})\phi(R^{0}_{i}).\eea
Thus, (\ref{eq3}) indicates that the corresponding analytical approximation to $T$ can be expressed as 
\[T^{\prime}=\max_{1\leq i\leq p}\tilde{J}_{i}(R^{0}_{i})  =\max_{1\leq i\leq p} \{2\Phi(R^{0}_{i})-1+ 2N^{-1}\hat{q}_{i}(R^{0}_{i})\phi(R^{0}_{i})\}.\]

The approximation $T^{\prime}$ is a better choice than $T$ since it is computationally more attractive. However, it may be difficult to obtain the  distribution for  $T^{\prime}$. One approach is to use  the bootstrap to estimate the CDF  $T^{\prime}$. To avoid the Monte Carlo approximation associated  the bootstrap technique, we consider \DNA{a transformation-based approach}. The idea is as follows.

\DNA{I}f  a monotone transformation of $T^{\prime}$ has a limiting distribution, then we can use that  transformation to develop the rejection region of the test. Due to the  monotonicity of the aforementioned transformation,  the inference based on $T^{\prime}$ or a monotone transformation of $T^{\prime}$ will be same.  Let $\Phi^{-1}(\cdot)$ denote the inverse function of the cdf of a standard normal distribution.  Theorem \ref{pro2} shows that
\begin{equation}
T^{\prime\prime}=[\max_{1\leq i\leq p} \Phi^{-1}(\tilde{J}_{i}(R^{0}_{i}))]^{2}-2\text{log}(p)+\text{log}{\text{log}(p)}
\end{equation}
has an extreme value distribution as $(n,p)\to \infty$. Thus, $T^{\prime\prime}$ is better choice than $T^{\prime}$  at least from the computational \DNA{cost perspective}. 

To obtain the limiting distribution of $T^{\prime\prime}$, we introduce two more assumptions which are similar to the assumptions of \textit{Lemma} 6 of \cite{4}.  Expression (\ref{a12}) in the Appendix shows that $\Phi^{-1}\{\tilde{J}_{i}(R^{0}_{i})\}=\Phi^{-1}\{J_{i}(R^{0}_{i})\}+o_{p}(N^{-1})$,  where $o_{p}(N^{-1})$ goes to $0$ in probability faster than $N^{-1}$. The probability integral transformation implies that under $H_0$, $\Phi^{-1}\{J_{i}(R^{0}_{i})\}\sim N(0,1)$ and further $(\Phi^{-1}\{J_{1}(R^{0}_{1})\},\ldots, \Phi^{-1}\{J_{p}(R^{0}_{p})\})^{\top}$ jointly follow as a $p$ dimensional multivariate normal distribution with zero mean vector and covariance matrix $\Gamma=(\gamma_{ij})$ (see \cite{20}). $\gamma_{ij}=\text{corr}(\Phi^{-1}\{J_{i}(R^{0}_{i})\}, \Phi^{-1}\{J_{j}(R^{0}_{j})\})$ is the Pearson correlation between $\Phi^{-1}\{J_{i}(R^{0}_{i})\}$ and $\Phi^{-1}\{J_{j}(R^{0}_{j})\}$. For  fixed $p$, $(\Phi^{-1}\{\tilde{J}_{1}(R^{0}_{1})\},\ldots, \Phi^{-1}\{\tilde{J}_{p}(R^{0}_{p})\})^{\top}$ converges weakly to   $(\Phi^{-1}\{J_{1}(R^{0}_{1})\},\ldots, \Phi^{-1}\{J_{p}(R^{0}_{p})\})^{\top}$ based on a multivariate version of Slutsky Theorem. Thus, to establish the the weak convergence of $T^{\prime\prime}$ when $N,p\to\infty$,  assumptions (A4) and (A5) are imposed on the $\Gamma$, the covariance matrix of the random vector $(\Phi^{-1}\{J_{1}(R^{0}_{1})\},\ldots, \Phi^{-1}\{J_{p}(R^{0}_{p})\})^{\top}$.
\begin{itemize}
\item[(A4)]  $\max_{i\leq i<j\leq p} |\gamma_{ij}|\leq r_{1}<1$  for some $0 \leq r_1<1$. (A4) is mild since $\max_{1\leq i<j\leq p}|\gamma_{ij}|=1$ would imply that the
$\Gamma$  is singular. 
\item[(A5)] $\max_{1\leq j\leq p}\sum_{i=1}^{p}\gamma^2_{ij} \leq r_{2}$ for some $r_2>0$. 
\end{itemize}


\begin{theorem}\label {pro2}
Let $\text{log}(p)=o(N^2)$. Then under the assumptions (A1)--(A5),  for any $x\in \mathbb{R}$,
\[ \text{P}\{T^{\prime\prime}\leq x\}\to \exp\{-\frac{1}{2\sqrt{\pi}}\exp(-x/2)\}\;\; \text{as}\;\;N\to \infty.\]
\end{theorem}
\noindent{ \textbf{Proof}} See Appendix

The limit distribution is type I extreme value distribution.  On the basis of the limiting null distribution, the approximate p-value is $1-\exp\{-\frac{1}{2\sqrt{\pi}}\exp(-t^{\prime\prime}/2)\}$ when $\text{log}(p)$ grows at a rate smaller than $N^2$. Based on Theorem \ref{pro2}, we can also construct asymptotically $\alpha$-level test that rejects the null hypothesis in \eqref{eq1} if $T^{\prime\prime}> q_{\alpha}$, where $q_{\alpha}= -2\log(2 \sqrt{\pi})-2 \log \log(1-\alpha)^{-1}$. 

Our test for small sample sizes and large dimension is in the same spirit as the test proposed by Cai et al. \cite{3}  for analysis of rare signals. Their test involves estimation of the inverse of $p\times p$ covariance matrices. Matrix inversion it is time-consuming $p$ \DNA{when} is large, whereas our test is computationally more efficient.  

We now analyse the power of the test. To study the asymptotic power, we consider a local alternative condition that the maximum absolute \DNA{value} of the standardized signals is \DNA{higher in order} than $\sqrt{\text{log}(p^2/\text{log}(p))},$ that is  
\bea \label{eq4}  \sqrt{\text{log}(p^2/\text{log}(p))}=o\biggl(\max_{1\leq j\leq p}\frac{|\mu_{xj}-\mu_{yj}|}{\sqrt{\sigma^2_{xj}/n+\sigma^2_{yj}/m}}\biggr).
\eea
Theorem \ref{pro3} shows that if (\ref{eq4}) holds, then  the power of the test, $T^{\prime\prime}>q_{\alpha}$, converges to 1  as $N,p\to\infty$.
\begin{theorem}\label{pro3}
Under the assumptions of Theorem \ref{pro2}, if $\bm{\mu}_x$ and $\bm{\mu}_y$ satisfy the condition in (\ref{eq4}), then
\[P(T^{\prime\prime}>q_{\alpha}|H_{1})\to 1.\]
\end{theorem}

\noindent{ \textbf{Proof}} See Appendix

Theorem \ref{pro3} indicates that the proposed test procedure should detect the discrepancy between two mean vectors $\bm{\mu}_{x}$, and $\bm{\mu}_{y}$ with probability tending probability to 1 if the maximum of the absolute standardized signals is of order higher than $\sqrt{\log (p^2/\log(p))}$ as $p, N\to\infty$. 
\DNA{In comparison,} Cai et al.'s test \cite{4}, \DNA{which is }based on \DNA{a} supremum-type test statistic, has an asymptotic power 1 provided the maximum of the absolute standardized signals is of higher order than $\sqrt{c \log(p^2)}$, for a certain unknown constant $c>0$.  \DNA{T}he asymptotic powers of sum-of-squares-type tests (Bai and Saranadasa \cite{7}; Srivastava and Du \cite{8}; Chen and Qin \cite{9}) converge to 1 \DNA{when} $Np^{-1/2}\|\bm{\mu}_{x} - \bm{\mu}_{y}\|^2_{2}\to\infty$  as $N, p\to\infty$.  The asymptotic power of GCT (Gregory et al. \cite{12}) converges to 1  provided $p^{-1/2}\sum\limits_{j=1}^{p}\frac{(\mu_{xj}-\mu_{yj})^2}{\sqrt{\sigma^2_{xj}/n +\sigma^2_{yj}/m}}\to \infty$ as $N, p\to\infty$. 

\DNA{To better understand the asymptotic power of the proposed test in comparison to the other tests}, consider \DNA{the following} scenario. \DNA{Suppose that all the signals are of } \DNA{equal} strength $\delta$, have unit variances and equal sample sizes for the two groups, $n=m$. Under the alternative, let $\mathcal{A}=\{i\in\{1,\ldots, p\}: \mu_{xi}-\mu_{yi}\neq 0\}$ denote the set of locations of the signals, and \DNA{let} the cardinality of $A$ \DNA{be} $p^r$, where $r\in(0,1)$. Under this scenario, power of the sum-of-squares-type tests  converge to 1 as $N,p\to \infty$ if $\sqrt{p}=o(2n\|\bm{\mu}_{x} - \bm{\mu}_{y}\|^2_{2})=o(p^{r}n\delta^2)$, that is if $\sqrt{p}/p^r=o(n\delta^2)$.   Similarly, the power of the GCT test converges to 1 if $\sqrt{p}/p^r=o(\sqrt{n}\delta^2)$, as $N,p\to \infty$. On the other hand, the power of our suggested test procedure goes to 1 as $N,p\to \infty$ if $\sqrt{\text{log}(p^2/\text{log}(p))}=o(\sqrt{n}\delta)$. Hence the power of sum-of-squares-type tests depend on how large $(\sqrt{n}\delta)^2$ compared to $p^{1/2-r}$. \DNA{In comparison}, the power of our \DNA{proposed} test depends on how large $\sqrt{n}\delta$ compared to $\sqrt{\text{log}(p^2/\text{log}(p))}$\DNA{, which has a much smaller rate of increase compared to $p^{1/2-r}$}. Under this particular example, the   sum-of-squares-type tests are less powerful when $r<1/2$ compared to $r>1/2$. Extending further, we may say that the our test procedure is \DNA{asymptotically} more powerful than the sum-of-squares-type tests with $r<1/2$.  



\section{Simulation Study}
\label{sec:simulation}

In this section we present empirical results based on simulated data. We compare \DNA{the empirical type I error rate and power of the proposed} test procedure, denoted by \textit{PREPR}, \DNA{under different models} with the \DNA{following} test methods: BS (Bai and Saranadasa \cite{7}), SD (Srivastava and DU \cite{8}),  CQ (Chen and Qin \cite{9}), CLX (Cai et al. \cite{4}), GM (Karl et al.\cite{12}), GL and GM (Karl et al.\cite{12}). GM and GL are  two variants of \cite{12}: moderate-p, and large-p.
Following \cite{12}, we chose the Parzen window with lag window size $2\sqrt{p}/3$ for both GM and GL. The number of variety of multivariate distributions and parameters are too large to allow a comprehensive, all-encompassing comparison. Hence we have chosen certain representative examples for illustration. We considered two setups to generate the data.

\subsection{Simulation setup  1}
\DNA{In setup 1, we considered two scenarios for the marginal distributions of the elements of the random vectors $\bm{X}$ and $\bm{Y}$. In} Scenario 1, each of the $p$ components of the random vectors $\bm{X}$ and $\bm{Y}$ follow a standard normal distribution. In Scenario 2, each of the $p$ components of the random vectors $X$ and $Y$ follow a centralized gamma(2,1) distribution. \DNA{To impose a dependence structure on joint distributions of $X$ and $Y$, we used a factor model as described in \cite{9}}. Simulation setup 1 considers variances for the all $p$ components of $\bm{X}$ and $\bm{Y}$. Under each of the above two scenarios, we consider three dependence models to  generate correlated observations: 
\begin{itemize}
\item[]\textbf{(M1)} Model 1 considers a moving average process of order 10. The moving coefficients vector was  a normalized vector of dimension 10, where the components of that vector were generated independently from Uniform (2, 3) and were kept fixed once generated through our simulations.
\item[]\textbf{(M2)} In Model 2, we consider long-range (LR) dependent structure that was also used in \cite{12}.  The self similarity parameter of the LR dependent structure was set at 0.625.  The algorithm to generate LR dependent samples can be found in \cite{21}.
\item[]\textbf{(M3)} In Model 3, we set $\Sigma_{x} = \Sigma_{y} = (1-\rho)\mathcal{I}_{p}+\rho\mathcal{J}_{p}$, where $\rho=0.4$ and $\mathcal{J}_p = \mathbf{1} \mathbf{1}^{\top}$ is the all ones matrix.
\end{itemize}

We considered three values for the dimension: $p=200, 1000$ and $3000$. For each combination of scenario (1 or 2), model (\textbf{M1}/\textbf{M2}/\textbf{M3}) and $p$, the empirical type I error and power were computed based on 2000 \DNA{randomly generated data sets}. The empirical powers were computed on against the numbers of signals $p\times r$, where we considered $r = 0.5\%, 1\%$ and $3\%$. \DNA{The parameter $r$ denotes the sparsity of the signal.} The sample sizes $(n, m)$ were \DNA{set} as $(35,35)$ and $(60, 40)$. In all the scenarios, we set  $\mu_{x}=(c_{p r}, 0_{p (1-r)})^{\top}$, where  $c_{p r} = \delta  (r/p, 2r/p, \ldots, 1)$ and $0_{p(1-r)}$ denotes a $ p(1-r)$-vector with entries equal to 0. In setup 1, we computed powers for $\delta=0.6, 0.9$. For the empirical type I error calculation, we assumed $\delta=0$. Without loss of generality, we set $\mu_{y}=0_{d\times 1}$, the mean vector of $(Y_{1},\ldots, Y_{d})^{\top}$ \DNA{as the zero vector for all simulations}.

\begin{table}[!htbp]
	\centering
	\caption{Empirical type I errors of PREPR and other test statistics at $5\%$ significance level. Data was generated under different models (M1/M2/M3) and the two scenarios under simulation setup 1. The dimension was set as $p=200$.}
	\begin{adjustbox}{width=0.6\textwidth}
		\begin{tabular}{l|cccccccc}
			%
			\hline
			\hline
			 & \multicolumn{7}{c}{Scenario 1} \\
			
			Model &\DNA{PREPR} & BS& SD & CQ & GM& GL& CLX& \\
			\hline
			\hline
			 & \multicolumn{7}{c}{$n_1=n_2=35$} \\
			\hline
			\hline
			M1& 0.039& 0.062& 0.038& 0.062& 0.165& 0.186& 0.061\\
			M2& 0.045& 0.056& 0.0507& 0.057& 0.073& 0.086& 0.085\\
			M3& 0.046& 0.053& 0.042& 0.053& 0.073& 0.085& 0.080\\
			\hline
			& \multicolumn{7}{c}{ $n_1=60,\;\; n_2=40$} \\
			\hline
		    M1& 0.041& 0.071& 0.049& 0.070& 0.161& 0.175& 0.069\\
			M2  & 0.047& 0.059& 0.056& 0.059& 0.067& 0.071& 0.073\\
			M3 & 0.048& 0.059& 0.047& 0.062& 0.063& 0.077& 0.070\\
			\hline\hline
			 & \multicolumn{7}{c}{Scenario 2} \\

			Model &\DNA{PREPR} & BS& SD & CQ & GM& GL&CLX& \\
			\hline\hline
			 & \multicolumn{7}{c}{$n_1=n_2=35$} \\
			\hline
			\hline
			M1& 0.038  &0.058& 0.045& 0.059& 0.164& 0.191& 0.056\\
			M2  &0.042&0.059& 0.047& 0.062& 0.067& 0.119& 0.067\\
			M3 & 0.044& 0.069& 0.054& 0.069& 0.068& 0.094& 0.070 \\
			\hline
			& \multicolumn{7}{c}{$n_1=60,\;\; n_2=40$} \\
			\hline
		    M1 & 0.045& 0.058& 0.041& 0.056& 0.148& 0.170& 0.065\\
			M2 & 0.036& 0.065& 0.048& 0.065& 0.063& 0.099& 0.066\\
			M3  & 0.044& 0.055& 0.046& 0.055& 0.067& 0.089& 0.075\\
			\hline

		\end{tabular}
	\end{adjustbox}
\label {table 1}	
\end{table}

\begin{table}[!htbp]
	\centering
	\caption{Empirical type I errors of the test statistics at $5\%$ significance level. Data was generated under simulation setup 1 and when the dimension was fixed as $p=1000$.}
	\begin{adjustbox}{width=0.6\textwidth}
		\begin{tabular}{l|cccccccc}
			%
			\hline
			\hline
			 & \multicolumn{7}{c}{Scenario 1} \\
			
			Model &\DNA{PREPR}& BS& SD & CQ & GM& GL& CLX& \\
			\hline
			\hline
			 & \multicolumn{7}{c}{$n_1=n_2=35$} \\
			\hline
			\hline
			M1& 0.047  & 0.058& 0.031 & 0.058 & 0.077& 0.087& 0.110\\
			M2& 0.049   & 0.058 & 0.031 & 0.057 & 0.086&0.062 & 0.116 \\
			M3& 0.051  & 0.049&0.034& 0.049& 0.070& 0.064& 0.109\\
			\hline
			& \multicolumn{7}{c}{ $n_1=60,\;\; n_2=40$} \\
			\hline
		    M1& 0.044  & 0.059& 0.035& 0.060 & 0.093&0.086 &0.084\\
			M2  & 0.043& 0.049 & 0.032 & 0.048 & 0.079&0.059&0.083\\
			M3 & 0.050 & 0.057 & 0.035 & 0.058 & 0.062& 0.048&0.097 \\
			\hline\hline
			 & \multicolumn{7}{c}{Scenario 2} \\

			Model &\DNA{PREPR}& BS& SD & CQ & GM& GL&CLX& \\
			\hline\hline
			 & \multicolumn{7}{c}{$n_1=n_2=35$} \\
			\hline
			\hline
			M1& 0.042  &0.064& 0.033 & 0.065 & 0.077& 0.101& 0.102\\
			M2  & 0.037 & 0.043 & 0.021 & 0.045 & 0.089&0.190 & 0.085 \\
			M3 & 0.031 & 0.056 & 0.026 & 0.059 & 0.133& 0.180& 0.097 \\
			\hline
			& \multicolumn{7}{c}{$n_1=60,\;\; n_2=40$} \\
			\hline
		    M1 & 0.044 & 0.062& 0.040& 0.062 & 0.082&0.097 & 0.088\\
			M2 & 0.040 & 0.054 & 0.033 & 0.054 & 0.083&0.128 & 0.088 \\
			M3  & 0.039 & 0.045 & 0.030 & 0.048 &0.073 & 0.094& 0.086\\
			\hline

		\end{tabular}
	\end{adjustbox}
\label {table 2}	
\end{table}
                 
\begin{table}[!htbp]
	\centering
	\caption{Empirical type I errors of the test statistics at $5\%$ significance level. Data was generated under simulation setup 1 and when the dimension was fixed as $p=3000$.}

	\begin{adjustbox}{width=0.6\textwidth}
		\begin{tabular}{l|cccccccc}
			%
			\hline
			\hline
			 & \multicolumn{7}{c}{Scenario 1} \\
			&\DNA{PREPR}& BS& SD & CQ & GM& GL& CLX& \\
			\hline
			\hline
			Model &\multicolumn{7}{c}{$n_1=n_2=35$} \\
			\hline
			\hline
			M1& 0.051  & 0.049& 0.023 & 0.049 & 0.069&0.067 & 0.142\\
			M2& 0.046   & 0.048 & 0.020 & 0.047 & 0.172&0.057 & 0.133 \\
			M3&0.050   & 0.056 &0.020 &0.055 &0.150&0.064  &0.131  \\
			
			\hline
			& \multicolumn{7}{c}{$n_1=60, \;\; n_2=40$} \\
			\hline
		    M1 & 0.045 & 0.049& 0.020& 0.0496 & 0.076 &0.063& 0.107\\
			M2 & 0.048 & 0.051 & 0.021 & 0.051 & 0.167&0.058 & 0.142 \\
			M3 &0.051  &0.053  &0.028  & 0.056 & 0.105 &0.056  &0.115\\
			\hline\hline
			 & \multicolumn{7}{c}{Scenario 2} \\
			
			Model &\DNA{PREPR}& BS& SD & CQ & GM& GL& CLX& \\
			\hline\hline
			& \multicolumn{7}{c}{$n_1=n_2=35$} \\
			\hline
			\hline
			M1 & 0.045 & 0.057& 0.024 & 0.057 & 0.078 &0.08& 0.120\\
			M2  & 0.032 & 0.044 & 0.015 & 0.046 & 0.164&0.147& 0.097\\
			M3 &0.037& 0.041& 0.012& 0.043& 0.129& 0.190& 0.1107\\
			\hline
			& \multicolumn{7}{c}{$n_1=60,\;\; n_2=40$} \\
			\hline
		    M1 & 0.041 & 0.049& 0.020& 0.049 & 0.051 &0.060& 0.100\\
			M2  & 0.044 & 0.052 & 0.023 & 0.053 & 0.165 &0.186& 0.107\\
			M3 &0.042& 0.047& 0.021& 0.048& 0.114& 0.116& 0.106\\
			\hline
			
		\end{tabular}
	\end{adjustbox}
	\label {table 3}	
\end{table}


\begin{table}[!htbp]
	\centering
	\caption{Empirical power of the test statistics at $5\%$ significance level. Data was generated under Scenario 1 within simulation setup 1 and the dimension was fixed as $p = 200$.}
	
	\resizebox{\textwidth}{!}{
		\begin{tabular}{lcc|*{7}{c}|*{7}{c}}

			\hline
			\hline
r & Model & $\delta$ &  \multicolumn{7}{c|}{$p = 200$} & \multicolumn{7}{c}{$p = 200$} \\
& && \DNA{PREPR} & BS& SD & CQ & GM& GL&CLX& \DNA{PREPR} & BS & SD & CQ & GM& GL&CLX\\
\hline
\hline
&&& \multicolumn{7}{c|}{$n_{1}=n_{2}=35$} &  \multicolumn{7}{c}{$n_{1}=60\;\; n_{2}=40$} \\
\hline
\hline
\multirow{3}{*}{$0.5\%$}  & M1 & \multirow{3}{*}{0.6}  &0.056& 0.066& 0.045& 0.067& 0.142&0.158& 0.093 &0.086&  0.092& 0.059& 0.100& 0.146& 0.149& 0.129 \\
& M2& & 0.138& 0.105& 0.093&0.106& 0.065& 0.049& 0.222&0.219& 0.147& 0.122& 0.146& 0.0645& 0.053& 0.294\\
 &M3& &0.122& 0.104& 0.076&0.103& 0.066& 0.064& 0.205&0.218& 0.126& 0.100& 0.124& 0.044& 0.044& 0.284\\
\hline
\multirow{3}{*}{$0.5\%$}  & M1  &   \multirow{3}{*}{0.9}  &0.142& 0.082& 0.061& 0.082& 0.135& 0.146& 0.213&0.271& 0.112& 0.074& 0.111& 0.116& 0.122& 0.400 \\
  & M2  &  &0.475&  0.221& 0.191& 0.221& 0.056& 0.039& 0.569&0.731& 0.313& 0.289& 0.314& 0.046& 0.033& 0.788\\
 &M3& & 0.460& 0.201& 0.165& 0.201& 0.059& 0.042& 0.560& 0.705& 0.279& 0.244& 0.280& 0.047& 0.034& 0.7660\\
\hline
\hline
&&& \multicolumn{7}{c|}{$n_{1}=n_{2}=35$} & \multicolumn{7}{c}{$n_{1}=60\;\; n_{2}=40$} \\
\hline
\hline
\multirow{3}{*}{$1\%$}&M1 & \multirow{3}{*}{0.6} & 0.069& 0.100& 0.073& 0.099& 0.132& 0.133& 0.117 & 0.114& 0.123& 0.093& 0.122& 0.122& 0.127& 0.165\\
&M2 & & 0.211& 0.266& 0.224& 0.2667& 0.083& 0.049& 0.319&0.372& 0.366& 0.340&0 0.365& 0.063& 0.041& 0.461\\
&M3&& 0.186& 0.217& 0.179& 0.219& 0.060& 0.041& 0.2725&0.333& 0.302& 0.272& 0.303& 0.047& 0.036& 0.420\\ 
\hline
\multirow{3}{*}{$1\%$} & M1 & \multirow{3}{*}{0.9} & 0.136 &0.098& 0.059& 0.098& 0.077& 0.064& 0.236&
0.327& 0.187& 0.138& 0.188&  0.082& 0.075& 0.402\\
& M2 &  & 0.704&  0.642& 0.599& 0.641& 0.049& 0.027& 0.786&0.887& 0.825& 0.792& 0.824& 0.033& 0.018& 0.927\\
&M3& & 0.628& 0.533& 0.482& 0.532& 0.049& 0.029& 0.726&0.850& 0.735& 0.699& 0.733& 0.028& 0.017& 0.908\\
\hline
\hline

&&& \multicolumn{7}{c|}{$n_{1}=n_{2}=35$} & \multicolumn{7}{c}{$n_{1}=60\;\; n_{2}=40$} \\ 
\hline
\hline
\multirow{3}{*}{$3\%$}  &  M1 & \multirow{3}{*}{0.6}  & 0.101& 0.190& 0.145& 0.190& 0.121& 0.119& 0.168 &0.171&  0.260& 0.208& 0.258& 0.126& 0.113& 0.229\\
&  M2 & & 0.375&  0.708& 0.669& 0.708& 0.153& 0.103& 0.509&0.630& 0.885& 0.864& 0.884& 0.129& 0.091& 0.725\\
&M3&&0.360& 0.654& 0.607& 0.654& 0.139& 0.089& 0.4945&0.568& 0.823& 0.794& 0.825& 0.125& 0.089& 0.664\\
\hline
\multirow{3}{*}{$3\%$}  &  M1 & \multirow{3}{*}{0.6}  & 0.334& 0.400&  0.332& 0.400& 0.127& 0.107& 0.436& 0.519&0.565& 0.494& 0.561& 0.147& 0.128& 0.598\\
&  M2 & & 0.916& 0.996& 0.993& 0.996& 0.123& 0.069& 0.969&0.993& 0.999& 0.999& 0.999& 0.069& 0.049& 0.997\\
&M3&&0.901& 0.992& 0.986& 0.992& 0.109& 0.074& 0.944& 0.985& 0.999& 0.999& 0.999& 0.089& 0.067& 0.993\\
\hline
\end{tabular}
	}
\label {table 4}
\end{table}   

\begin{table}[!htbp]
	\centering
	\caption{Empirical power of the test statistics at $5\%$ significance level. Data was generated under Scenario 2 within simulation setup 1 and the dimension was fixed as $p = 200$}
	
	\resizebox{\textwidth}{!}{
		\begin{tabular}{lcc|*{7}{c}|*{7}{c}}

			\hline
			\hline
r & Model & $\delta$ &  \multicolumn{7}{c|}{$p = 200$} & \multicolumn{7}{c}{$p = 200$} \\
& && \DNA{PREPR} & BS& SD & CQ & GM& GL&CLX& \DNA{PREPR} &BS& SD & CQ & GM& GL&CLX\\
\hline
\hline
&&& \multicolumn{7}{c|}{$n_{1}=n_{2}=35$} &  \multicolumn{7}{c}{$n_{1}=60\;\; n_{2}=40$} \\
\hline
\hline
\multirow{3}{*}{$0.5\%$}  & M1 & \multirow{3}{*}{0.6}  &0.050 & 0.077& 0.060& 0.077& 0.157& 0.178& 0.088&0.083& 0.079& 0.059& 0.078& 0.135& 0.153& 0.127\\
& M2& & 0.138& 0.0987& 0.085& 0.104& 0.058& 0.052& 0.225&0.264& 0.131& 0.127& 0.138& 0.071& 0.0507& 0.335\\
 &M3& &0.124& 0.095& 0.082& 0.100& 0.058& 0.064& 0.199&0.250& 0.135& 0.108& 0.138& 0.056& 0.055& 0.326\\
\hline
\multirow{3}{*}{$0.5\%$}  & M1  &   \multirow{3}{*}{0.9}  &0.161& 0.091& 0.063& 0.091& 0.122& 0.136& 0.231&0.276& 0.104& 0.075& 0.105& 0.115& 0.118& 0.346\\
  & M2  &  &0.498& 0.240& 0.207& 0.251& 0.065& 0.0373& 0.595&0.700& 0.290& 0.263& 0.293&0.0507& 0.0287& 0.782\\
 &M3& &0.473& 0.193& 0.177& 0.200& 0.056& 0.033& 0.580& 0.692& 0.267& 0.234& 0.269& 0.044& 0.027& 0.769\\
\hline
\hline
&&& \multicolumn{7}{c|}{$n_{1}=n_{2}=35$} & \multicolumn{7}{c}{$n_{1}=60\;\; n_{2}=40$} \\
\hline
\hline
\multirow{3}{*}{$1\%$}&M1 & \multirow{3}{*}{0.6} & 0.213& 0.159& 0.111& 0.160& 0.104& 0.103& 0.290&0.107& 0.123& 0.086& 0.125& 0.117& 0.126& 0.160\\
&M2 & & 0.192& 0.251& 0.222& 0.263& 0.0767& 0.037& 0.300&0.432& 0.375& 0.369& 0.383& 0.085& 0.037& 0.534\\
&M3&& 0.186& 0.217& 0.179& 0.219& 0.060& 0.041& 0.273&0.382& 0.308& 0.267& 0.312& 0.057& 0.026& 0.465\\ 
\hline
\multirow{3}{*}{$1\%$} & M1 & \multirow{3}{*}{0.9} & 0.136 &0.098& 0.059& 0.098& 0.077& 0.064& 0.236&0.361& 0.192& 0.151& 0.195& 0.089& 0.081& 0.440\\
& M2 &  & 0.674&0.611& 0.591& 0.619& 0.043& 0.0127& 0.787&0.890& 0.816& 0.800& 0.817& 0.024& 0.0087& 0.934\\
&M3& & 0.624& 0.537& 0.513& 0.545& 0.044& 0.014& 0.753& 0.869& 0.757& 0.728& 0.756& 0.025& 0.011& 0.912\\
\hline
\hline

&&& \multicolumn{7}{c|}{$n_{1}=n_{2}=35$} & \multicolumn{7}{c}{$n_{1}=60\;\; n_{2}=40$} \\ 
\hline
\hline
\multirow{3}{*}{$3\%$}  &  M1 & \multirow{3}{*}{0.6}  & 0.107& 0.182& 0.145& 0.184& 0.139& 0.134& 0.171& 0.211& 0.242& 0.187& 0.241& 0.123& 0.110& 0.274\\
&  M2 & & 0.384& 0.710& 0.696& 0.715& 0.146& 0.056& 0.565&0.690&0.893& 0.879& 0.892& 0.147& 0.067& 0.787\\
&M3&&0.371& 0.635& 0.607& 0.639& 0.124& 0.054& 0.531&0.666& 0.829& 0.811& 0.829& 0.139& 0.078& 0.762\\
\hline
\multirow{3}{*}{$3\%$}  &  M1 & \multirow{3}{*}{0.9}  & 0.327& 0.388& 0.320& 0.390& 0.122& 0.090& 0.443&0.545& 0.557& 0.488& 0.557& 0.146& 0.123& 0.634\\
&  M2 & & 0.228& 0.403& 0.313& 0.404& 0.169& 0.053&0.398&0.299&0.711&0.553&0.711&0.290&0.095&0.555\\
&M3&&0.922& 0.993& 0.995& 0.994& 0.115& 0.039& 0.968& 0.993&1.000&1.000&1.000&0.061& 0.033& 0.998\\
\hline
\end{tabular}
	}
\label {table 5}
\end{table}   

\begin{table}[!htbp]
	\centering
	\caption{Empirical power of the test statistics at $5\%$ significance level. Data was generated under Scenario 1 within simulation setup 1 and the dimension was fixed as $p = 1000$ and $p = 3000$ respectively.}
	
	\resizebox{\textwidth}{!}{
		\begin{tabular}{lcc|*{7}{c}|*{7}{c}}

			\hline
			\hline
r & Model & $\delta$ &  \multicolumn{7}{c|}{$p = 1000$} & \multicolumn{7}{c}{$p = 3000$} \\
& && \DNA{PREPR} & BS& SD & CQ & GM& GL&CLX& \DNA{PREPR}    &BS& SD & CQ & GM& GL&CLX\\
\hline
\hline
&&& \multicolumn{7}{c|}{$n_{1}=n_{2}=35$} &  \multicolumn{7}{c}{$n_{1}=n_{2}=35$} \\
\hline
\hline
\multirow{3}{*}{$0.5\%$}  & M1 & \multirow{3}{*}{0.6}  &0.059 &0.067& 0.039& 0.067& 0.074& 0.078& 0.126&0.061&0.062&0.025&0.063&0.074&0.044&0.166\\
  & M2& & 0.095& 0.096& 0.064& 0.097& 0.121& 0.051&0.202&0.114&0.112&0.051&0.112&0.235&0.046&0.258\\
 &M3& &0.104& 0.081& 0.047& 0.081& 0.083& 0.047& 0.209&0.106&0.109& 0.049& 0.108& 0.194&0.041& 0.264 \\

\hline
\multirow{3}{*}{$0.5\%$}  & M1  &   \multirow{3}{*}{0.9}  &0.105 & 0.085& 0.047& 0.085& 0.081& 0.078& 0.2079&0.116&0.086&0.037&0.086&0.095&0.057&0.252\\
  & M2  &  &0.113 & 0.079& 0.049& 0.078& 0.081& 0.077&0.215 &0.391&0.244&0.125&0.245&0.245&0.037&0.670\\
 &M3& & 0.449& 0.255& 0.180& 0.256& 0.094& 0.032& 0.640&0.401& 0.227& 0.109& 0.227& 0.210& 0.040& 0.655\\
\hline
\hline
&&& \multicolumn{7}{c|}{$n_{1}=n_{2}=35$} & \multicolumn{7}{c}{$n_{1}=n_{2}=35$} \\
\hline
\hline

\multirow{3}{*}{$1\%$}&M1 & \multirow{3}{*}{0.6} & 0.060& 0.074& 0.041& 0.073& 0.082& 0.083&0.128&0.072&0.085&0.036&0.084&0.095&0.060&0.179\\
&M2 & & 0.120& 0.136& 0.090& 0.137& 0.131& 0.053&0.247&0.152&0.204&0.105&0.205&0.267&0.041&0.339\\
&M3&& 0.124&0.117& 0.074& 0.116& 0.094& 0.042& 0.230&0.163&0.188&0.086&0.188&0.219&0.05&0.342\\ 
\hline
\multirow{3}{*}{$1\%$} & M1 & \multirow{3}{*}{0.9} & 0.136 &0.098& 0.059& 0.098& 0.077& 0.064& 0.236&0.159 & 0.135&0.065&0.134&0.111&0.050&0.303\\
& M2 &  & 0.463& 0.308& 0.228& 0.308& 0.119& 0.036& 0.670&0.574 & 0.534&0.375&0.536&0.193&0.027&0.838\\
&M3& & 0.447& 0.266& 0.194& 0.265& 0.097& 0.034& 0.624&0.562&0.457&0.293&0.459&0.169&0.029&0.809\\
\hline
\hline

&&& \multicolumn{7}{c|}{$n_{1}=n_{2}=35$} & \multicolumn{7}{c}{$n_{1}=n_{2}=35$} \\ 
\hline
\hline
\multirow{3}{*}{$3\%$}  &  M1 & \multirow{3}{*}{0.6}  & 0.082  & 0.106& 0.065& 0.106& 0.080& 0.065& 0.168&0.092 &0.156&0.071&0.157&0.134&0.050&0.210\\
&  M2 & & 0.228& 0.403& 0.313& 0.404& 0.169& 0.053&0.398&0.299&0.711&0.553&0.711&0.290&0.095&0.555\\
&M3&& 0.215& 0.334& 0.245& 0.335& 0.139& 0.0545& 0.392&0.323&0.633&0.462&0.634&0.361&0.085&0.546\\

\hline
\multirow{3}{*}{$3\%$}  & M1 & \multirow{3}{*}{0.9}  & 0.225  & 0.215& 0.144& 0.214& 0.088& 0.056&0.357&0.289 &0.378&0.224&0.380&0.184&0.052&0.503 \\
& M2 & & 0.769& 0.880& 0.821& 0.882& 0.084&0.023&0.917 &0.893&0.997&0.991&0.997&0.248&0.046&0.990\\
&M3& & 0.746& 0.809& 0.732& 0.809& 0.095& 0.032& 0.887& 0.884& 0.996& 0.980& 0.996& 0.240& 0.057& 0.988\\
\hline
\hline
&&& \multicolumn{7}{c|}{$n_{1}=60\;\;n_{2}=40$} &  \multicolumn{7}{c}{$n_{1}=60\;\;n_{2}=40$} \\
\hline
\hline
\multirow{3}{*}{$0.5\%$}&M1 & \multirow{3}{*}{0.6}& 0.062 &0.075& 0.051& 0.076& 0.080& 0.077&0.117&0.075&0.079&0.048&0.076&0.079&0.076&0.137\\
 &M2& & 0.167& 0.111& 0.084& 0.113& 0.100& 0.055& 0.271&0.190&0.149&0.090&0.150&0.195&0.052&0.347\\
 &M3& & 0.141& 0.103& 0.071& 0.102& 0.087& 0.049& 0.245&0.171& 0.140& 0.078& 0.140& 0.163& 0.053& 0.318\\  
\hline
\multirow{3}{*}{$0.5\%$}&M1&\multirow{3}{*}{0.9}& 0.180 & 0.088& 0.058& 0.087& 0.072& 0.070& 0.281&0.184 &0.101&0.052&0.100&0.081&0.055&0.316\\
&M2&& 0.536 & 0.229& 0.176& 0.227& 0.091& 0.040&0.754&0.639 &0.374&0.251&0.373&0.160&0.036&0.885\\
&M3& &0.689&0.389& 0.317& 0.390& 0.086& 0.033& 0.839&0.630& 0.310& 0.193& 0.307&0.142&0.046& 0.853\\
\hline
\hline
&&& \multicolumn{7}{c|}{$n_{1}=60\;\;n_{2}=40$} &  \multicolumn{7}{c}{$n_{1}=60\;\;n_{2}=40$} \\ 
\hline
\hline

\multirow{3}{*}{$1\%$}&M1 & \multirow{3}{*}{0.6}& 0.075  & 0.082& 0.052& 0.083& 0.077& 0.076& 0.132&0.089&0.096&0.054&0.096&0.089&0.059&0.175\\
 &M2& & 0.224& 0.189& 0.144& 0.188& 0.114& 0.050&0.347&0.254&0.166&0.101&0.167&0.098&0.048&0.406\\
&M3&& 0.196& 0.164& 0.122& 0.162& 0.100& 0.047& 0.310&0.274&0.250&0.160&0.249&0.175&0.049&0.431 

\\
\hline
\multirow{3}{*}{$1\%$}&M1 & \multirow{3}{*}{0.9}& 0.224  & 0.117& 0.079& 0.114& 0.069& 0.059&0.323&0.254 &0.166&0.101&0.167&0.099&0.048&0.406\\
&M2& & 0.709& 0.447& 0.371& 0.448&  0.068& 0.026&0.875 &0.840  &0.767&0.660&0.768&0.0882&0.016&0.976\\
&M3& &0.678& 0.386& 0.325& 0.386& 0.072& 0.034& 0.837&0.833&0.681&0.549&0.675&0.091&0.025&0.961\\
\hline
\hline

&&&  \multicolumn{7}{c|}{ $n_{1}=60\;\;n_{2}=40$} & \multicolumn{7}{c}{$n_{1}=60\;\;n_{2}=40$}\\ 
\hline
\hline
\multirow{3}{*}{$3\%$}&M1 & \multirow{3}{*}{0.6} & 0.101  & 0.134& 0.092& 0.134& 0.073& 0.059& 0.182&0.137 &0.211&0.130&0.212&0.130&0.063&0.249\\
&M2& & 0.422& 0.582& 0.502& 0.581& 0.116& 0.050& 0.596&0.534 &0.889&0.817&0.889&0.407&0.089&0.757\\
&M3& & 0.376& 0.493& 0.423& 0.495& 0.117& 0.057&0.531&0.522& 0.865& 0.761& 0.862& 0.279& 0.094& 0.747\\
\hline
\multirow{3}{*}{$3\%$}&M1 & \multirow{3}{*}{0.9} &0.382&0.303&0.229&0.303&0.086&0.055& 0.507&0.893&0.997& 0.991& 0.997& 0.248& 0.046&0.990\\
 &M2& &0.946&0.981&0.967&0.980&0.0304&0.011&0.989&0.991&1.000&1.000&1.000&0.109&0.025&1.000\\
 &M3&& 0.943& 0.964& 0.943& 0.963& 0.038& 0.016& 0.984&0.990& 1.000& 1.000& 1.000& 0.124& 0.030& 1.000\\
\hline
\end{tabular}
	}
\label {table 6}
\end{table}

\begin{table}[!htbp]
	\centering
	\caption{Empirical power of the test statistics at $5\%$ significance level. Data was generated under Scenario 2 within simulation setup 1 and the dimension was fixed as $p = 1000$ and $p = 3000$ respectively.}
	
	\resizebox{\textwidth}{!}{
		\begin{tabular}{lcc|*{7}{c}|*{7}{c}}

			\hline
			\hline
r & Model & $\delta$ &  \multicolumn{7}{c|}{$p = 1000$} & \multicolumn{7}{c}{$p = 3000$} \\
& && \DNA{PREPR} & BS& SD & CQ & GM& GL&CLX& \DNA{PREPR} &BS& SD & CQ & GM& GL&CLX\\
\hline
\hline
&&& \multicolumn{7}{c|}{$n_{1}=n_{2}=35$} &  \multicolumn{7}{c}{$n_{1}=n_{2}=35$} \\
\hline
\hline
\multirow{3}{*}{$0.5\%$} &M1 & \multirow{3}{*}{0.6}& 0.058&0.060& 0.034& 0.061& 0.073& 0.095& 0.111&0.068&0.073& 0.031 &0.074& 0.089& 0.087&0.132\\
 &M2& & 0.080& 0.086& 0.054& 0.091& 0.113& 0.100&0.185& 0.077 &0.109& 0.046& 0.114& 0.241& 0.191&0.207\\
 &M3& & 0.078& 0.086& 0.051& 0.089& 0.097& 0.074& 0.177&0.089&0.097& 0.041& 0.099& 0.183& 0.094& 0.228 \\
  
\hline
\multirow{3}{*}{$0.5\%$} & M1 & \multirow{3}{*}{0.9} & 0.105 & 0.080& 0.044& 0.082& 0.082& 0.080& 0.209&0.109&0.100& 0.040& 0.101& 0.096& 0.061&0.238\\
 & M2 & & 0.325 & 0.156& 0.107& 0.164& 0.112& 0.042& 0.556&0.376 &0.233& 0.133& 0.244& 0.227& 0.047&0.688\\
 & M3& & 0.315& 0.139& 0.099& 0.147& 0.101& 0.035& 0.524&0.375&0.208& 0.107& 0.212& 0.203& 0.024& 0.652\\ 
\hline
\hline
&&& \multicolumn{7}{c}{$n_{1}=n_{2}=35$} & \multicolumn{7}{c}{$n_{1}=n_{2}=35$} \\
\hline
\hline

\multirow{3}{*}{$1\%$} & M1 & \multirow{3}{*}{0.6} & 0.059& 0.078& 0.045& 0.081& 0.079& 0.086&0.117&0.062&0.080&0.030&0.080&0.092&0.056&10.141\\
 &M2& & 0.108& 0.129& 0.088& 0.134& 0.123& 0.066&0.227 &0.140&0.212&0.115&0.217&0.287&0.067&0.328\\
 &M3& &0.111& 0.111& 0.075& 0.115& 0.096& 0.050& 0.229&0.134& 0.158& 0.072& 0.164& 0.202& 0.035& 0.309\\
\hline
\multirow{3}{*}{$1\%$}&M1 & \multirow{3}{*}{0.9}& 0.133 &0.100& 0.059& 0.101& 0.081& 0.073&0.243& 0.146 &0.123& 0.059& 0.123& 0.112& 0.050 &0.305\\
 &M2& & 0.448&0.294& 0.234& 0.305& 0.115& 0.013&0.698 & 0.574&0.530& 0.370& 0.543& 0.157& 0.035&0.876\\
 &M3& &0.446& 0.261& 0.192& 0.270& 0.114& 0.017&0.652&0.540& 0.455& 0.308& 0.461& 0.186& 0.004& 0.820\\
\hline
\hline
&&& \multicolumn{7}{c}{$n_{1}=n_{2}=35$} & \multicolumn{7}{c}{$n_{1}=n_{2}=35$} \\
\hline
\hline
\multirow{3}{*}{$3\%$}&M1 & \multirow{3}{*}{0.6} & 0.079  & 0.110& 0.063& 0.112& 0.078& 0.065&0.158& 0.287& 0.389& 0.249& 0.392& 0.194& 0.056&0.509\\
 &M2& & 0.221  &0.394 & 0.324& 0.405& 0.167& 0.014& 0.424 &0.279  &0.700&0.559&0.710&0.226&0.055&0.584\\
 &M3& & 0.208& 0.308& 0.239& 0.312& 0.134&0.018& 0.385&0.283& 0.621& 0.473& 0.630& 0.359& 0.015& 0.567\\
\hline
\multirow{3}{*}{$3\%$}&M1 & \multirow{3}{*}{0.9} & 0.219  & 0.215& 0.148& 0.217& 0.094& 0.054&0.349 &0.070& 0.151&0.070&0.156&0.133&0.056&0.183\\
 &M2& & 0.765& 0.875& 0.840& 0.880& 0.074& 0.058&0.933 &0.884&0.996&0.992&0.996&0.381&0.047&0.994 \\
 &M3& &0.764& 0.801& 0.753& 0.809& 0.099& 0.009& 0.917&0.888& 0.996& 0.990& 0.996& 0.251& 0.005& 0.994\\
\hline
\hline
&&& \multicolumn{7}{c}{$n_{1}=60\;\;n_{2}=40$} & \multicolumn{7}{c}{$n_{1}=60\;\;n_{2}=40$} \\
\hline
\hline
\multirow{3}{*}{$0.5\%$} &M1 & \multirow{3}{*}{0.6} & 0.064 &0.067& 0.042& 0.066& 0.077& 0.084&0.121&0.088&0.071&0.044&0.071&0.079&0.066&10.46\\
 &M2& & 0.191& 0.114& 0.082& 0.118&0.116& 0.05&0.3147&0.233&0.146&0.082&0.153&0.232&0.056&0.415\\
 &M3& &0.200& 0.087& 0.063& 0.090& 0.075& 0.044& 0.311& 0.206&0.119&0.063&0.120& 0.172& 0.041& 0.376\\
  
\hline
\multirow{3}{*}{$0.5\%$} &M1& \multirow{3}{*}{0.9} & 0.184 & 0.087& 0.058& 0.090& 0.071& 0.069&0.291&0.202&0.095&0.049&0.093&0.079&0.051&0.334\\
 &M2& & 0.543 & 0.232& 0.190& 0.237& 0.112& 0.021&0.787&0.662&0.361&0.258&0.369&0.176&0.033&0.913\\
 &M3& &0.542& 0.214& 0.173& 0.214& 0.102& 0.024& 0.744&0.614& 0.299& 0.200& 0.301& 0.163& 0.009& 0.869\\

\hline
\hline
&&& \multicolumn{7}{c}{$n_{1}=60\;\;n_{2}=40$} & \multicolumn{7}{c}{$n_{1}=60\;\;n_{2}=40$}  \\ 
\hline
\hline

\multirow{3}{*}{$1\%$} &M1 & \multirow{3}{*}{0.6} & 0.082  & 0.079& 0.053& 0.078& 0.079& 0.084&0.144&0.103&0.095&0.047&0.094&0.085&0.053&0.188\\
 &M2& & 0.253& 0.174& 0.134& 0.181& 0.120& 0.030&0.409 &0.331&0.280&0.185&0.290&0.223&0.011&0.568\\
 &M3& &0.235& 0.145& 0.110& 0.150& 0.091& 0.031& 0.363&0.318& 0.249& 0.172& 0.253& 0.212& 0.013& 0.517\\

\hline
\multirow{3}{*}{$1\%$} &M1 & \multirow{3}{*}{0.9}& 0.256  & 0.117& 0.082& 0.118& 0.073& 0.065 &0.360&0.281&0.169&0.098&0.166&0.106&0.044&0.434\\
 &M2& & 0.708& 0.450& 0.400& 0.457&  0.079& 0.010&0.889&0.983&0.774&0.689&0.781&0.106&0.027&0.983\\
 &M3& & 0.693& 0.387& 0.326& 0.393& 0.065& 0.013& 0.870&0.836& 0.677& 0.567& 0.683& 0.110& 0.004& 0.971\\
\hline
\hline
&&& \multicolumn{7}{c}{$n_{1}=60\;\;n_{2}=40$} & \multicolumn{7}{c}{$n_{1}=60\;\;n_{2}=40$}\\ 
\hline
\hline
\multirow{3}{*}{$3\%$} &M1 & \multirow{3}{*}{0.6}& 0.131  & 0.139& 0.098& 0.139& 0.074& 0.058 &0.217&0.161&0.205&0.120&0.207&0.126&0.052&0.290\\
 &M2& & 0.502& 0.557& 0.504& 0.562& 0.129& 0.013&0.685 & 0.665&0.894&0.840&0.893&0.329&0.046&0.870\\
 &M3& &0.478& 0.476& 0.407& 0.480& 0.116& 0.016& 0.640&0.632& 0.847& 0.763& 0.848& 0.299&0.033& 0.854\\
\hline
\multirow{3}{*}{$3\%$} &M1 & \multirow{3}{*}{0.9} &0.418&0.303&0.230&0.306&0.087&0.0485 &0.545&0.576&0.570&0.443&0.571&0.186&0.065&0.736 \\
 &M2&  &0.950&0.977&0.971&0.978&0.030&0.011&0.994 &0.994&1.000&1.000&1.000&0.122&0.060&1.000\\
 &M3&  &  0.951& 0.955& 0.934& 0.955& 0.043& 0.010& 0.995&0.997& 1.000& 1.000& 1.000&0.132& 0.008& 1.000\\
\hline
\end{tabular}
	}
\label {table 7}
\end{table}   
\subsection{Simulation setup 2} 
In Simulation setup 2, we use the same setup as in Simulation setup 1, except that we now use covariance matrices with unequal variances are described in \textbf{(M4)} and \textbf{(M5)}.
\begin{itemize}
\item[]\textbf{(M4)} Model 4 considers $\Sigma_{x}=\Sigma_{y}=\text{D}^{1/2}\Sigma^{*}\text{D}^{1/2}$, where D is a diagonal matrix with with diagonal elements $d_{ii}$=Unif(1,5), $\Sigma^{*}=(\sigma^{*}_{ij})$, $\sigma^{*}_{ii}=1$, and $\sigma^{*}_{ij}=|i-j|^{-5/2}$ if $i\neq j.$  Such a
covariance matrix was also used in \cite{3}. \item[] \textbf{(M5)} In Model 5, $\Sigma_{x}=\Sigma_{y}=\text{D}^{1/2}\Sigma^{*}\text{D}^{1/2}$, where D is a diagonal matrix with  diagonal elements $d_{ii}$=Unif(1,5), and $\Sigma^{*}$ has a long-range dependence structure with diagonal that are all ones. $\Sigma^{*}$ has the similar  structure as the the \textbf{(M2)} in Simulation setup 1. \DNA{describe the structure}
\end{itemize}

\begin{table}[!htbp]
	\centering
	\caption{Empirical type I errors of the test statistics at $5\%$ significance level. Data was generated under simulation setup 2 and models (M4/M5) and the two scenarios. Dimension was fixed at $p=200$}
	\begin{adjustbox}{width=0.6\textwidth}
		\begin{tabular}{l|cccccccc}
			%
			\hline
			\hline
			 & \multicolumn{7}{c}{Scenario 1} \\
			
			Model &\DNA{PREPR} & BS& SD & CQ & GM& GL& CLX& \\
			\hline
			\hline
			 & \multicolumn{7}{c}{$n_1=n_2=35$} \\
			\hline
			\hline
			M4& 0.048& 0.054& 0.046& 0.054& 0.065& 0.074& 0.089\\
			M5& 0.033& 0.076& 0.030& 0.074& 0.072& 0.073& 0.073\\

			\hline
			& \multicolumn{7}{c}{ $n_1=60,\;\; n_2=40$} \\
			\hline
		    M4& 0.044& 0.060& 0.050& 0.063& 0.065& 0.072& 0.072\\
			M5  & 0.040& 0.069& 0.025& 0.069& 0.059& 0.069& 0.068\\
			
			\hline\hline
			 & \multicolumn{7}{c}{Scenario 2} \\

			Model &\DNA{PREPR}& BS& SD & CQ & GM& GL&CLX& \\
			\hline\hline
			 & \multicolumn{7}{c}{$n_1=n_2=35$} \\
			\hline
			\hline
			M4& 0.039& 0.056& 0.040& 0.056& 0.061&0.099& 0.065\\
			M5  &0.032& 0.070& 0.029& 0.070& 0.071& 0.116& 0.062\\
			\hline
			& \multicolumn{7}{c}{$n_1=60,\;\; n_2=40$} \\
			\hline
		    M4 & 0.041& 0.065& 0.049& 0.067& 0.066& 0.097& 0.077\\
			M5 & 0.039& 0.078& 0.024& 0.078& 0.070& 0.106& 0.066\\
			\hline
		\end{tabular}
	\end{adjustbox}
\label {table 8}	
\end{table}

\begin{table}[!htbp]
	\centering
	\caption{Empirical type I errors of the test statistics at $5\%$ significance level. Data was generated under simulation setup 2 and models (M4/M5) and the two scenarios. Dimension was fixed at $p=1000$.}
	\begin{adjustbox}{width=0.6\textwidth}
		\begin{tabular}{l|cccccccc}
			%
			\hline
			\hline
			 & \multicolumn{7}{c}{Scenario 1} \\
			
			Model &\DNA{PREPR} & BS& SD & CQ & GM& GL& CLX& \\
			\hline
			\hline
			 & \multicolumn{7}{c}{$n_1=n_2=35$} \\
			\hline
			\hline
			M4&0.055 &0.058& 0.034& 0.058& 0.078& 0.059& 0.115\\
			M5& 0.044& 0.051& 0.032& 0.052& 0.087& 0.059& 0.112\\

			\hline
			& \multicolumn{7}{c}{ $n_1=60,\;\; n_2=40$} \\
			\hline
		    M4& 0.045& 0.054& 0.038& 0.054& 0.086& 0.088& 0.087\\
			M5  & 0.047& 0.051& 0.037& 0.048& 0.071& 0.057& 0.094\\
			
			\hline\hline
			 & \multicolumn{7}{c}{Scenario 2} \\

			Model & \DNA{PREPR} & BS& SD & CQ & GM& GL&CLX& \\
			\hline\hline
			 & \multicolumn{7}{c}{$n_1=n_2=35$} \\
			\hline
			\hline
			M4& 0.030& 0.066& 0.031& 0.066& 0.087& 0.110& 0.089\\
			M5  &0.034& 0.047& 0.032& 0.049& 0.085& 0.188& 0.074\\
			\hline
			& \multicolumn{7}{c}{$n_1=60,\;\; n_2=40$} \\
			\hline
		    M4 & 0.039& 0.060& 0.031& 0.057& 0.080& 0.095& 0.072\\
			M5 & 0.040& 0.059& 0.035& 0.058& 0.0887& 0.107& 0.101\\
			
			\hline

		\end{tabular}
	\end{adjustbox}
\label {table 9}	
\end{table}

\begin{table}[!htbp]
	\centering
	\caption{Empirical powers of the test statistics at $5\%$ significance level. Data was generated under simulation setup 2 for scenario 1 and two models (M4/M5). Dimension was fixed at $p = 200$ and $p = 1000$.}
	
	\resizebox{\textwidth}{!}{
		\begin{tabular}{lcc|*{7}{c}|*{7}{c}}

			\hline
			\hline
r & Model & $\delta$ &  \multicolumn{7}{c|}{$p = 200$} & \multicolumn{7}{c}{$p = 1000$} \\
& && \DNA{PREPR} & BS& SD & CQ & GM& GL&CLX& \DNA{PREPR} &BS& SD & CQ & GM& GL&CLX\\
\hline
\hline
&&& \multicolumn{7}{c|}{$n_{1}=n_{2}=35$} &  \multicolumn{7}{c}{$n_{1}=n_{2}=35$} \\
\hline
\hline
\multirow{3}{*}{$0.2$}  & M4 & \multirow{3}{*}{1.5}  &0.280& 0.165& 0.126& 0.165& 0.052& 0.041& 0.364&0.301 & 0.107& 0.067& 0.106& 0.098& 0.051& 0.479\\
  & M5& & 0.263& 0.076& 0.023& 0.076& 0.582& 0.606& 0.338&0.276& 0.074& 0.016& 0.074& 0.862& 0.874& 0.451\\

\hline
\multirow{3}{*}{$0.4$}  & M4  &   \multirow{3}{*}{1.5}  &0.751& 0.464& 0.487& 0.464& 0.042& 0.024& 0.830&0.694& 0.341& 0.337& 0.341& 0.088& 0.028& 0.934\\
  & M5  &  &0.727& 0.102& 0.032& 0.102& 0.256& 0.254& 0.798&0.653 &0.082&0.013& 0.082& 0.690& 0.718& 0.908\\
 
\hline

\multirow{3}{*}{$0.6$}&M4 & \multirow{3}{*}{1.5} & 0.909& 0.968& 0.973& 0.968& 0.116& 0.076& 0.956&0.992 &0.996& 1.000& 0.996& 0.158& 0.074& 1.000\\
&M5 & & 0.831& 0.213& 0.085& 0.213& 0.253& 0.239& 0.892&0.944& 0.109& 0.031& 0.109& 0.284& 0.268& 0.989\\

\hline
\hline
&&& \multicolumn{7}{c|}{$n_{1}=60\;\;n_{2}=40$} &  \multicolumn{7}{c}{$n_{1}=60\;\;n_{2}=40$} \\
\hline
\hline
\multirow{3}{*}{$0.5\%$}&M4 & \multirow{3}{*}{1.5}&0.480& 0.226& 0.173& 0.226& 0.050& 0.038& 0.555&0.449& 0.111& 0.080& 0.112& 0.070& 0.036& 0.710\\
 &M5& & 0.441& 0.088& 0.035& 0.090& 0.518& 0.546& 0.527&0.447& 0.081& 0.016& 0.080& 0.864& 0.870& 0.701\\
 
\hline
\multirow{3}{*}{$1\%$}&M4&\multirow{3}{*}{1.5}&0.923& 0.651& 0.701& 0.650& 0.023& 0.017& 0.954&0.840& 0.531 &0.559& 0.529& 0.053& 0.021& 0.989\\
&M5&& 0.927& 0.121& 0.045& 0.121& 0.213& 0.205& 0.954&0.800& 0.092& 0.015& 0.091& 0.552& 0.590& 0.984\\

\hline
\multirow{3}{*}{$03\%$}&M4 & \multirow{3}{*}{1.5}&0.988& 0.998& 0.999& 0.998& 0.068& 0.051& 0.995& 1.000& 1.000& 1.000& 1.000 &0.076& 0.040& 1.000\\
 &M5& & 0.956& 0.246& 0.095& 0.251& 0.226& 0.212& 0.973&0.992 &0.131& 0.030& 0.135& 0.265& 0.258& 0.999\\
\hline
\end{tabular}
	}
\label {table 10}
\end{table}

\begin{table}[!htbp]
	\centering
	\caption{Empirical powers of the test statistics at $5\%$ significance level. Data was generated under simulation setup 2 for scenario 2 and two models (M4/M5). Dimension was fixed at $p = 200$ and $p = 1000$.}
	
	\resizebox{\textwidth}{!}{
		\begin{tabular}{lcc|*{7}{c}|*{7}{c}}

			\hline
			\hline
r & Model & $\delta$ &  \multicolumn{7}{c|}{$p = 200$} & \multicolumn{7}{c}{$p = 1000$} \\
& && \DNA{PREPR} & BS& SD & CQ & GM& GL&CLX& \DNA{PREPR} &BS& SD & CQ & GM& GL&CLX\\
\hline
\hline
&&& \multicolumn{7}{c|}{$n_{1}=n_{2}=35$} &  \multicolumn{7}{c}{$n_{1}=n_{2}=35$} \\
\hline
\hline
\multirow{3}{*}{$0.5\%$}  & M4 & \multirow{3}{*}{1.5}  &0.252& 0.166& 0.136& 0.173& 0.059& 0.042& 0.3615&0.291& 0.107& 0.077& 0.113& 0.108& 0.050& 0.508\\
  & M5& & 0.250& 0.085& 0.030& 0.084& 0.615& 0.650& 0.343&0.288& 0.075& 0.013& 0.076& 0.875& 0.882& 0.472\\

\hline
\multirow{3}{*}{$1\%$}  & M4  &   \multirow{3}{*}{1.5}  &0.753 &0.451& 0.506& 0.460& 0.040& 0.019& 0.841&0.645& 0.340& 0.354& 0.351& 0.090& 0.009& 0.934\\
  & M5  &  &0.717& 0.109& 0.045& 0.109& 0.241& 0.236& 0.804& 0.638& 0.070& 0.012& 0.070 &0.676& 0.726& 0.901\\

\hline

\multirow{3}{*}{$3\%$}&M4 & \multirow{3}{*}{1.5} & 0.918& 0.961& 0.977& 0.963& 0.089& 0.044& 0.966&0.994& 0.996& 1.000& 0.996& 0.146& 0.020& 1.000\\
&M5 & &0.825& 0.194 &0.074& 0.194& 0.246 &0.225 &0.893
&0.949& 0.125& 0.02758& 0.125& 0.299& 0.272& 0.988\\

\hline
\hline
&&& \multicolumn{7}{c|}{$n_{1}=60\;\;n_{2}=40$} &  \multicolumn{7}{c}{$n_{1}=60\;\;n_{2}=40$} \\
\hline
\hline
\multirow{3}{*}{$0.5\%$}&M4 & \multirow{3}{*}{1.5}&0.457& 0.223& 0.186& 0.235& 0.048& 0.030& 0.551&0.440& 0.140& 0.119& 0.142& 0.117& 0.035& 0.700\\
 &M5& & 0.467& 0.083& 0.028& 0.083& 0.519& 0.549& 0.543&0.430& 0.073& 0.013& 0.075& 0.868& 0.872& 0.680\\

\hline
\multirow{3}{*}{$1\%$}&M4&\multirow{3}{*}{1.5}&0.925& 0.661& 0.715& 0.667& 0.020& 0.011& 0.955&0.864& 0.516& 0.567& 0.520& 0.055& 0.011& 0.984\\
&M5&& 0.895& 0.106& 0.041& 0.111& 0.210& 0.199& 0.936&0.793& 0.072& 0.014& 0.071& 0.542& 0.591& 0.977\\

\hline
\multirow{3}{*}{$3\%$}&M4 & \multirow{3}{*}{1.5}&0.991& 0.997& 0.999& 0.997& 0.063& 0.034& 0.998& 1.000& 1.000& 1.000& 1.000 &0.068& 0.015& 1.000\\
 &M5& & 0.958& 0.245& 0.093& 0.255& 0.219& 0.208& 0.977&0.995& 0.129& 0.029& 0.132& 0.256& 0.238& 0.999\\

\hline
\end{tabular}
	}
\label {table 11}
\end{table}  

\subsection{Discussion the results from the simulation }
Results from Simulation setup 1 presented in Tables \ref{table 1}-\ref{table 3} display the empirical type I errors of the new test \textit{PREPR} \DNA{along} with the tests BS, CQ, GL, GM, and CLX.  These tables show that  empirical type I errors of new test \textit{PREPR}  are close the nominal size 0.05 under the all scenarios and all the models, and  they range from 0.032 to 0.051. The type I errors of tests BS and CQ are comparable. In some cases, they are  liberal with values ranging  from 0.041--0.071 and 0.043--0.070 respectively. Test SD seems to be conservative with increase $p$, and its empirical type-I errors vary from 0.012 to 0.05. Whereas the tests CLX, GL and GM are generally liberal under all scenarios and models with their empirical type I errors ranging from 0.056--0.142, 0.059--0.191 and 0.051--0.165 respectively. \DNA{Overall}, in all of the considered scenarios and models under the simulation setup 1, the proposed test \textit{PREPR} outperforms BS, CQ, SD, GM, GL and CLX in terms of size control as indicated by its better type I error accuracy.

Empirical type I errors corresponding to the Simulation setup 2 are presented in Tables \ref{table 8}-\ref{table 9}. These tables show that type I errors of these tests are consistent with their  empirical type I errors reported in Simulation setup 1 
For example,  the type I errors of \textit{PREPR} range from 0.034--0.055, whereas those of BS,CQ, SD, GM, GL, and CLX range from 0.047--0.078, 0.048--0.078, 0.024--0.050, 0.059--0.088, 0.057--0.116 and 0.066--0.115, respectively.

The empirical powers of different tests under Simulation setup 1 are presented in Tables \ref{table 4}-\ref{table 7}. Tables \ref{table 4} and  \ref{table 5} report the powers for $p=200$, and tables \ref{table 6} and \ref{table 7} report the powers for $p=1000$ and $p=3000$ respectively. The tables show that \textit{PREPR} has more power overall than the tests BS, SD, CQ, GM, and GL when \DNA{the signals are sparse - } $r=0.5\%$, and $r=1\%$. As the signals become less sparse at $r=3\%$, the tests BS, SD, and CQ are more powerful than \textit{PREPR}, particularly when \DNA{the signals are moderately strong with} $\delta=0.6$. Though tests GM and GL are liberal, tables \ref{table 4}-\ref{table 7} display that GM and GL have the smallest power in most cases. The test CLX enjoys the maximum  power in all considered  scenarios and models under the Simulation setup 1, irrespective of whether signals are sparse or dense. Such performance of CLX is not surprising since it showed inflated type I errors. Tables \ref{table 10}-\ref{table 11} display empirical powers corresponding to Simulation setup 2. From these tables, we can observe that comparisons are similar to the cases that are reported Simulation setup 1. 

In summary, based on the numerical results we have performed that \textit{PREPR} shows better control on the type I error than the remaining tests, and is more powerful than sum-of-squares tests BS, SD, CQ, GL, and GM  against the sparse alternatives. 

\section{Data Analysis/Application}
\label{sec:realdata} 
We considered an application of \text{PREPR} test for analyzing  gene expression data in terms of gene sets. A gene-set analysis refers to a statistical analysis to identify whether some functionally pre-defined sets of genes express differently under different experimental conditions. A gene-set analysis is the beginning for biologists to understand the patterns in the differential expression. There are two major categories of a gene-set test: competitive gene-set test and self-contained gene-set test (\cite{22}). The first one tests a set of genes, say $\mathcal{G}$, of interest with complementary set of genes which are not in $\mathcal{G}$. Self-contained gene-set test considers the null hypothesis that none of the genes in $\mathcal{G}$ are differentially expressed. The proposed two sample test \textit{PREPR} is applicable to  the self-contained gene-set test. We apply \textit{PREPR} test on the following two data sets to perform the self-contained gene-set test.

\subsection{ Colon tumor and normal data}
We consider the application of \textit{PREPR} test method for analyzing the colon tumor and normal data set consisting of gene expression levels of different cell types.  The data set has 40 colon tumor samples and 22 normal colon tissue samples. The data contains the expression of $p = 2000$ genes with highest minimal intensity across the 62 samples; see (\cite{23}) for more details. This data set is available in the R-package ``\texttt{plsgenomics}". We are interested to test the null hypothesis that the tumor group has the same gene expression levels as the normal group. Figure \ref{fig1} shows that distributions of many gene expression values for both groups are highly skewed and hence we log-transformed the data prior to analysis. Before we apply \textit{PREPR} method along with other methods to this data set, it is worth to investigate whether these test procedures control type I error for this data set. Due to unavailability of
repeated data sets on colon tumor tissues and normal tissues, we can apply similar idea as \cite{25} to compare performance of different tests with respect to type I error. We randomly partition 40 samples of colon tumor data into two groups where each group has 20 samples. Since both groups have the same distribution, we can expect that both groups have the same mean vector. For each partition, we apply the tests \textit{PREPR}, BS, SD, CQ, GM,  GL, CLX and calculate corresponding p-values. We repeat this procedure 1000 times for obtaining type I error for each test corresponding to the colon tumor  data. We also randomly partition 22 samples in the normal data into two groups where each group has 11 samples, and similarly obtain type I error for each test corresponding to the normal data. Table \ref{table:12}  reports empirical type I error of different tests at $5\%$ nominal level corresponding to normal data and colon tumor data, respectively. Table \ref{table:12} shows that in normal data \textit{PREPR} achieves type I error closest to the nominal level of 0.05 whereas the other methods show difficulty controlling type I error. In tumor data, \textit{PREPR} displays better control of type I error compared to  the tests CQ, SD, GL, and GM. We apply these tests to the complete data sets: $n_1=22$ samples from normal data and $n_2=40$ from colon tumor data. 

We obtain the p-values for the tests \textit{PREPR}, BS, SD, CQ, GM, GL, and CLX under the null hypothesis, the tumor group has the same gene expression levels as the normal group, are 0.004, 0.002, 0.252, 0.092, 0.000, 0.000 and $1.87 \times 10^{-6}$ respectively. Though tests BS, GM, GL, and CLX the reject the null hypothesis that no genes  are differentially expressed, it is worth noting from Table \ref{table:12} that these tests have inflated type I error for both the colon tumor and normal data.

\begin{figure}[!htb]
\centering
\includegraphics[height=12cm,width=15cm]{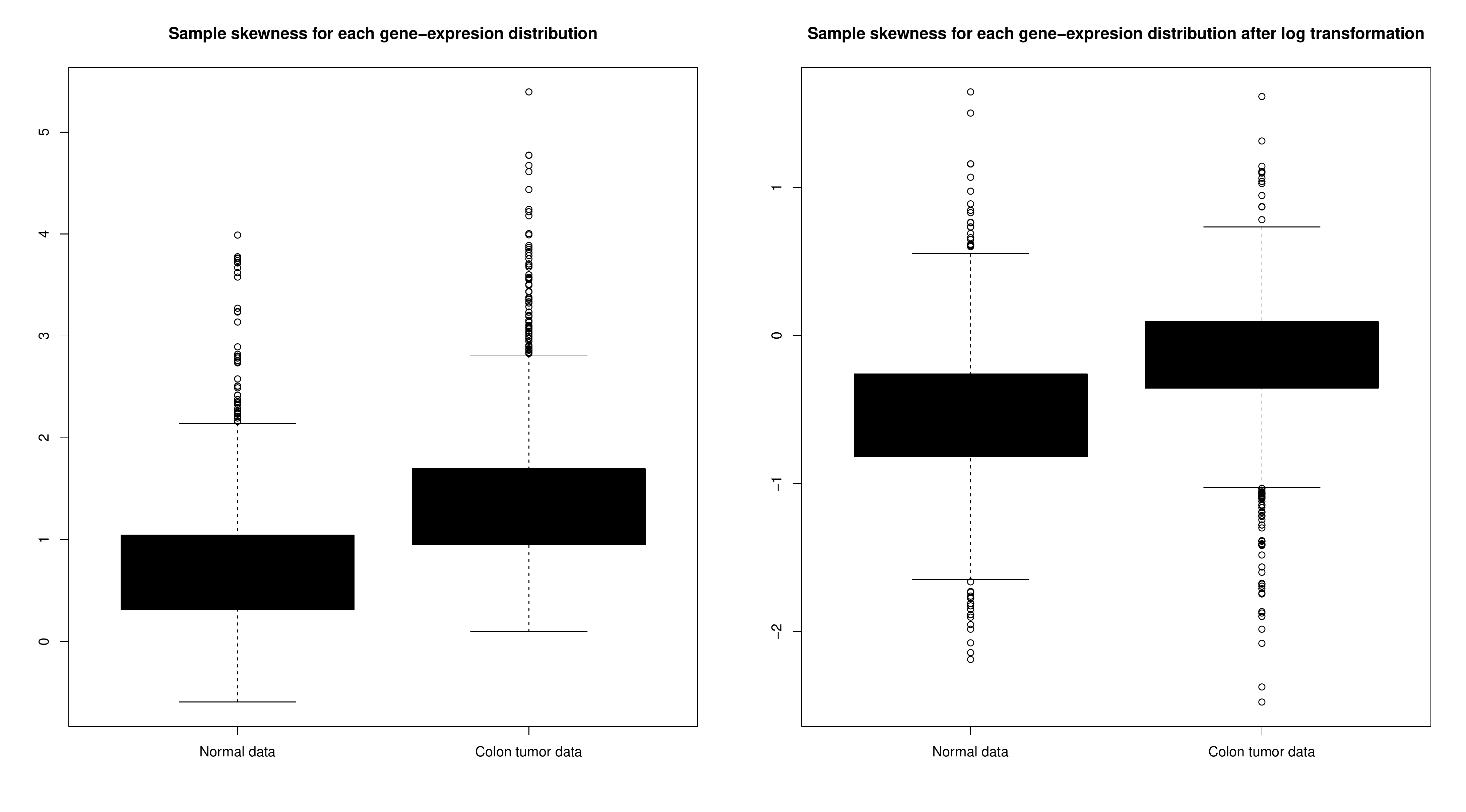}
\caption{Sample skewness of the distributions of $2000$ gene expressions from the colon tumor data set. The plots display the skewness before (left) and after (right) applying a log-transformation to the data. }
\label{fig1}
\end{figure}

\begin{table}[!htbp]
\centering
\caption{Empirical type I error of different tests at $5\%$ significance level computed from 1000 randomly partitioned data sets constructed using the normal and colon tumor data, respectively}
\begin{adjustbox}{width=0.7\textwidth}
\begin{tabular}{ccc}
\hline
\text{Test Method} &normal group& colon tumor group\\
\hline
  PREPR  &0.035&0.022\\
  BS  &0.087&0.066\\
  SD &0.009&0.007\\
  CQ  &0.010&0.011\\
  GM  &0.949&0.931\\
  GL &0.946&0.926\\
  CLX  &0.112&0.040\\[7pt]  
\hline
\end{tabular}
\end{adjustbox}
\label{table:12}
\end{table}

\subsection{Embryonal tumors  data}

Embryonal tumors are most commonly occurring brain tumors in children.  The major subtypes of Embryonal tumors of the  central nervous system in children are medulloblastoma, atypical teratoid rhabdoid tumor (ATRT), and embryonal  tumours; see \cite{24} more details. We applied the different tests to  determine whether the medulloblastoma and  ATRT  could be molecularly distinguished. Medulloblastoma  and  ATRT tumors have expression of $p=5597$ genes on 10 subjects, respectively. This data set is available at \href{http://www.stat.cmu.edu/~jiashun/Research/software/GenomicsData/Brain/}{http://www.stat.cmu.edu/~jiashun/Research/software/GenomicsData/Brain/}. Similar to the colon tumor and normal data set, we first computed the empirical type I error rates of tests \textit{PREPR}, BS, SD, CQ, GM, GL, and CLX separately by randomly partitioning the medulloblastoma tumor data and  ATRT tumor data. To compute the empirical type I errors for each data set, we randomly  partitioned each data set into  two  groups where  each  group  had 5 samples. Then we applied a similar scheme as the colon tumor and normal data to calculate the type I errors. These results, reported in Table \ref{table:13}, show that tests \textit{PREPR}, BS, and CQ give good control on type I error for both tumor data whereas GM, GL, and CLX exhibit overly inflated type I errors. 

The p-values of the tests \textit{PREPR}, BS, SD, CQ, GM, GL and CLX under the null hypothesis that  Medulloblastoma  and  ATRT tumors have the same gene expression levels, are $6.23 \times 10^{-06}$, $3.99 \times 10^{-11}$, 0.411, $3.91 \times 10^{-11}$, 0.000, 0.000 and 0.000, respectively. Though tests  GM, GL, and CLX the reject the null hypothesis that no genes  are differentially expressed, Table \ref{table:13} shows that these tests have inflated type I error, indicating that these results are not very reliable.
\begin{table}[!htbp]
\centering
\caption{Empirical type I error of different tests at $5\%$ significance level computed from 1000 randomly partitioned data sets constructed using the medulloblastoma and ATRT data sets respectively.}
\begin{adjustbox}{width=0.7\textwidth}
\begin{tabular}{ccc}
\hline
\text{Test Method} & medulloblastoma tumor & ATRT tumor\\
\hline
  PREPR &0.045&0.040\\
  BS  &0.058&0.056\\
  SD &0.010&0.006\\
  CQ  &0.058&0.056\\
  GM  &0.883&0.884\\
  GL &0.982&0.966\\
  CLX  &1.000&0.990\\[7pt]  
\hline
\end{tabular}
\end{adjustbox}
\label{table:13}
\end{table}

\section{Discussion}
\label{sec:discussion}
In this article, we developed the test \text{PREPR} for two-sample mean vectors testing problem when dimension of data  is larger than the sample size. \text{PREPR} is based on the prepivoting technique \DNA{to achieve better sampling distribution of the root}. We derived the limiting null distribution of the test statistic, which is asymptotically pivotal. We analyzed the power of the suggested \text{PREPR} test statistic and it is a consistent test provided the maximum absolute value of the standardized signals is higher in order than $\sqrt{\text{log}(p^2/\text{log}(p))}$. In addition to these theoretical investigations, the simulation studies presented in Section 3 demonstrate that the \text{PREPR} can control the type I error regardless of whether the distributions of the two samples are normal or non-normal. In terms of type I error control, \textit{PREPR} provides better control than the considered existing tests: BS, CQ, SD, CLX, GM and GL in many cases. Further, the simulation studies also demonstrate that \textit{PREPR} attains more power than BS, CQ, SD,  GM, and GL when there is a small number of signals, i.e. components of the random vectors with non-zero means difference. We also observed in the simulation studies that CLX always has more power \textit{PREPR} is due to fact CLX enjoys inflated type-I error. \DNA{Application of \textit{PREPR} to the two data sets also indicate that the proposed test controls type I error for gene expression data sets. It is also capable of detecting differential expression in gene sets under different experimental conditions.} For the simplicity and the computational efficacy, \textit{PREPR} can be successfully applied in practice. \DNA{We have developed an \texttt{R} package to implement \textit{PREPR}. The package is available at \href{https://github.com/dnayyala/prepr}{https://github.com/dnayyala/prepr}.}

\pagebreak

\appendix

\renewcommand{\thesection}{A\Roman{section}}
\numberwithin{equation}{section}

\section*{Appendix}
\label{sec:appendix}
\setcounter{equation}{0}
\textbf{ Proof of Theorem  \ref{pro1}}: The proof of Theorem \ref{pro1}  is available with the attached supplement.\newline\newline
\textbf{ Proof of Theorem  \ref{pro2}}: To sketch the proof of Theorem  \ref{pro2}, we need Lemmas \ref{le1} and \ref{le2}. Below is Lemma \ref{le1}.
\begin{lemma} \label{le1}
Let $\mu_{xi,l}=\text{E}(X_{i}-\mu_{xi})^{l}<\infty$ and $\mu_{yi,l}=\text{E}(Y_{i}-\mu_{yi})^{l}<\infty$ denote the $l$th order central moment of the $i$ th component of $\bm{X}$ and $\bm{Y}$,   $i = 1, \ldots, p$ and $l\geq 1$.  If the corresponding estimates are defined as
$\hat{\mu}_{xi,l}=n^{-1}\sum_{j=1}^{n}(X_{j,i}-\bar{X}_i)^{l}$ and $\hat{\mu}_{yi,l}=m^{-1}\sum_{j=1}^{m}(Y_{j,i}-\bar{Y}_i)^{l}$,
then 
\begin{equation}
\max\limits_{i\geq 1} (\hat{\mu}_{xi,l}-\mu_{xi,l})=O_{p}(N^{-1/2}),\;\;\text{and} \;\; \max_{i\geq 1}(\hat{\mu}_{yi,l}-\mu_{yi,l})=O_{p}(N^{-1/2}).
\end{equation}
\end{lemma}

\begin{proof}: Define
\[p_{i,n}=P\{|\mu_{xi,l}-\mu_{xi,l}|>cn^{-1/2}\},\;\; \text{and}\;\; \delta_{i,n}=1-p_{i,n}.\]
Then for every $i$ and $u\geq n$
\[ \inf_{v\geq n}\delta_{i,v}\leq \delta_{i,u}.\]
Consequently, for  every $r$ and $u\geq n$
\[ \sum_{i=1}^{r}\inf_{v\geq n}\delta_{i,v}\leq \sum_{i=1}^{r}\delta_{i,u},\]
which indicates that for every $r$ and $n$
\[ \sum_{i=1}^{r}\inf_{v\geq n}\delta_{i,v}\leq \inf_{u\geq n} \sum_{i=1}^{r}\delta_{i,u}.\]
The above inequality implies that
\[\sum_{i=1}^{r}\inf_{v\geq n}(1-p_{i,v})\leq \inf_{u\geq n} \sum_{i=1}^{r}(1-p_{i,u}) \implies
\sum_{i=1}^{r}\sup_{v\geq n}p_{i,v}\geq \sup_{u\geq n} \sum_{i=1}^{r}p_{i,u}\]
Thus for every $r$, \[\sup_{u\geq n} \sum_{i=1}^{r}p_{i,u}\leq \sum_{i=1}^{\infty}\sup_{v\geq n}p_{i,v}\implies  \lim_{n\to\infty}\sup_{u\geq n} \sum_{i=1}^{r}p_{i,u}\leq \sum_{i=1}^{\infty}\lim_{n\to\infty }\sup_{v\geq n}p_{i,v},\]
consequently \[\lim_{n\to\infty}\sup_{u\geq n} \sum_{i=1}^{\infty}p_{i,u}\leq \sum_{i=1}^{\infty}\lim_{n\to\infty }\sup_{v\geq n}p_{i,v}.\] 
Since the last inequality holds for every $c$, we have
\bea \label{a9}
\lim_{c\to\infty} \limsup_{n\to\infty}\sum_{i=1}^{\infty}P\{|\hat{\mu}_{xi,l}-\mu_{xi,l}|>cn^{-1/2}\}&\leq& \sum_{i=1}^{\infty}\lim_{c\to\infty} \limsup_{n\to\infty }P\{|\hat{\mu}_{xi,l}-\mu_{xi,l}|>cn^{-1/2}\}\nonumber\\
\eea
Now,
\bea \label{a10}  
\lim_{c\to\infty} \limsup_{n\to\infty} P\{\max_{i\geq 1 }|\hat{\mu}_{xi,l}-\mu_{xi,l}|>cn^{-1/2}\} &\leq& \lim_{c\to\infty} \limsup_{n\to\infty} \sum_{i=1}^{\infty} P\{|\hat{\mu}_{xi,l}-\mu_{xi,l}|>cn^{-1/2}\}\nonumber\\
&\leq&\sum_{i=1}^{\infty}\lim_{c\to\infty} \limsup_{n\to\infty }P\{|\hat{\mu}_{xi,l}-\mu_{xi,l}|>cn^{-1/2}\}\nonumber,\\
\eea
the last line follows form (\ref{a9}). Since $\hat{\mu}_{xi,l}-\mu_{xi,l}=O_{p}(n^{-1/2})$ (see pp. 72, \cite{25}), so \[\lim_{c\to\infty} \limsup_{n\to\infty }P\{|\hat{\mu}_{xi,l}-\mu_{xi,l}|>cn^{-1/2}\} \to 0,\] and (\ref{a10}) implies that
\[\lim_{c\to\infty} \limsup_{n\to\infty} P\{\max_{1\leq i<\infty }|\hat{\mu}_{xi,l}-\mu_{xi,l}|>cn^{-1/2}\}\to 0.\]

Thus, we have  
\[ \max_{1\leq i <\infty}|\hat{\mu}_{xi,l}-\mu_{xi,l}|=O_{p}(n^{-1/2}) = O_p(N^{-1/2}),\] because $n/N=O(1)$.
Similarly, it can be shown that \[ \max_{1\leq i <\infty}(\hat{\mu}_{yi,l}-\mu_{yi,l})=O_{p}(N^{-1/2}).\] 
\end{proof}

\begin{lemma}\label{le2}
Under assumption (A1)-(A3),
\[\max_{1\leq i\leq p}\Phi^{-1}[\hat{J}_{i}(R^{0}_{i})]=\max_{1\leq i\leq p} Z_{i}+o_{p}(N^{-1}),\]
holds uniformly in $i$, where $Z_i \sim N(0,1)$ for each $i=1,\ldots, p.$
\end{lemma}
\begin{proof}
Lemma \ref{le1} indicates that 
$\hat{\sigma}^2_{xi}=\sigma^2_{xi}+O_{p}(N^{-1/2})$,
$\hat{\sigma}^2_{yi}=\sigma^2_{yi}+O_{p}(N^{-1/2})$,
$\hat{\gamma}_{yi}=\gamma_{yi}+O_{p}(N^{-1/2})$,
$\hat{\gamma}_{xi}=\gamma_{xi}+O_{p}(N^{-1/2})$, $\hat{\gamma}_{yi}=\gamma_{yi}+O_{p}(N^{-1/2})$, $\hat{\kappa}_{xi}=\gamma_{xi}+O_{p}(N^{-1/2})$,
 $\hat{\kappa}_{yi}=\kappa_{yi}+O_{p}(N^{-1/2})$ hold  uniformly in $i$, and these facts imply that
 \[ \hat{q}_i(x)=q_{i}(x)+O_{p}(N^{-1/2}),\]
hold uniformly in $i$ since $\hat{q}_i(x)$ are function of the above sample quantities. Thus from Theorem \ref{pro1} and (\ref{eq3}) we have 
 \bea \label{a11}
 \tilde{J}_{i}(R^{0}_i)&=&2\Phi(R^{i}_{0})-1+2N^{-1}[q_{i}(R^{0}_i)+O_{p}(N^{-1/2})]\nonumber\\
 &=&2\Phi(R^{i}_{0})-1+2N^{-1}q_{i}(R^{0}_i)+O_{p}(N^{-3/2})\nonumber\\
 &=&J_{i}(R^{i}_{0})+o_{p}(N^{-1}),
 \eea
holds uniformly in $i$.

Applying a Taylor series expansion on $\Phi^{-1}(x)$ and using \eqref{a11}, we can derive
\bea \label{a12}
 \Phi^{-1}(\tilde{J}_{i}(R^{0}_i))&=&\Phi^{-1}(J_{i}(R^{i}_{0})))+o_{p}(N^{-1}))\nonumber\\
 &=& \Phi^{-1}((J_{i}(R^{i}_{0})))+o_{p}(N^{-1})\nonumber\\
 &=&Z_{i}+o_{p}(N^{-1})
 \eea
 where $Z_i=\Phi^{-1}((J_{i}(R^{i}_{0})))$.  \ref{a12} indicates that the order of sizes of the differences between $ \Phi^{-1}(\tilde{J}_{i}(R^{0}_i))$ and $Z_i$ same across all $i$. hence, it is safe to conclude 
\ben
\max_{1\leq i\leq p}\Phi^{-1}[\hat{J}_{i}(R^{0}_{i})]=\max_{1\leq i\leq p} [Z_i+o_{p}(N^{-1})]=\max_{1\leq i\leq p} Z_i+o_{p}(N^{-1}).
\een

\end{proof}

\begin{lemma}\label{le3}
If $\text{log}(p)=o(N^2)$ then
\begin{equation*}
N^{-1}\max_{1\leq i\leq p}Z_i\to 0\;\; \text{in probability as}\;\; N\to\infty.
\end{equation*}
\end{lemma}
\begin{proof}
By Chebyshev's inequality
\[P\{N^{-1}\max_{1\leq i\leq p}Z_i > \epsilon\}\leq \frac{\text{E}(\max\limits_{1\leq i \leq p}|Z_i|)}{N\epsilon},\]
for any $\epsilon>0.$ Now,
\ben
\text{E}(\max_{1\leq \leq p}Z_i)&=& t^{-1}\text{E}\biggl\{\text{log}e^{t\max\limits_{1\leq i\leq p}|Z_i|}\biggr\}\\
&\leq& t^{-1} \text{E}\biggl\{\text{log}\sum_{i=1}^{p}e^{t|Z_i|}\biggr\}\\
&\leq& t^{-1}\biggl\{\text{log}\biggl[\sum_{i=1}^{p}\text{E}(e^{tZ_i})+\sum_{i=1}^{p}\text{E}(e^{-tZ_i})\biggr]\biggr\},
\een 
the last two expressions can be obtained by applying $e^{|x|}\leq e^{x}+e^{-x}$ and Jensen's inequality respectively.
For any $t \in \mathbb{R}$, E$(e^{tZ_i}) = e^{t^2/2}$ since $Z_{i}\sim N(0,1)$. Hence, we have
\bea \label{a13}
\text{E}(\max_{1\leq \leq p}Z_i)&=&t^{-1}\biggl\{\text{log}\biggl[\sum_{i=1}^{p}\text{E}(e^{tZ_i})+\sum_{i=1}^{p}\text{E}(e^{-tZ_i})\biggr]\biggr\}\nonumber\\
&=&t^{-1}\text{log}(2pe^{t^2/2})\nonumber\\
&=&\frac{\text{log} (2p)}{t}+\frac{t}{2}.
\eea
Since $\frac{\text{log} (2p)}{t}+\frac{t}{2}$ attains maximum value at $t=\sqrt{2\text{log} (2p)}$, the upper bound in \eqref{a13} will be
\[\text{E}(\max_{1\leq \leq p}Z_i)\leq \frac{\text{log} (2p)}{\sqrt{2\text{log} (2p)}}+\frac{\sqrt{2\text{log} (2p)}}{2}=\sqrt{2\text{log} (2p)}.\] 
We now have
\[P\{N^{-1}\max_{1\leq i\leq p}Z_i>\epsilon\}\leq \frac{\sqrt{2\text{log} (2p)}}{\epsilon N}.\]

The above inequality indicates Lemma \ref{le3} because  $\frac{\sqrt{2\text{log} (2p)}}{\epsilon N}\to 0$ under the assumption $\text{log}(2p)=o(N^2)$. 
\end{proof}

\textbf{Proof of Theorem \ref{pro2}}:
We can express $T^{\prime\prime}$ as follows
\bea \label{a14}
T^{\prime\prime}&=&\{ \max_{1\leq i\leq p} Z_{i}+o_{p}(N^{-1})\}^{2}-2\text{ln}(p)+\text{ln}(\text{ln}(p))\nonumber\\
&=& \left (\max_{1\leq i\leq p}Z_i \right)^2+2o_{p}(N^{-1})\max_{1\leq i\leq p}Z_{i}+o_{p}(N^{-2})-2\text{ln}(p)+\text{ln}(\text{ln}(p)).
\eea
Using Lemma \ref{le3}, we we can write the expression (\ref{a14}) as
\[T^{\prime\prime}= (\max_{1\leq i\leq p}Z_{i})^{2}-2\text{ln}(p)+\text{ln}(\text{ln}(p))+o_{p}(1).\]
The result in pp. 309 of \cite{20} indicates $(Z_1,\ldots, Z_p)^{\top}\sim N_{p}(\bm{0},\Gamma)$, where \[\gamma_{ij}=\text{corr}(\Phi^{-1}\{J_{i}(R^{0}_i)\},\Phi^{-1}\{J_{j}(R^{0}_j)\] is the $(i,j)$the element of $\Gamma$.
Now, the assumptions (A4), (A5) and \textit{Lemma} 6 of Cai et al. \cite{4} conclude that \[P\biggl\{(\max_{1\leq i\leq p}Z_{i})^{2}-2\text{ln}(p)+\text{ln}(\text{ln}(p))\leq x)\biggr\}\to \exp\{-\frac{1}{2\sqrt{\pi}}\exp(-x/2)\},\]
as $p\to\infty.$ Thus $T^{\prime\prime}$ converges in distribution to random variable whose cdf is \[\exp\{-\frac{1}{2\sqrt{\pi}}\exp(-x/2)\},\] as $N,p\to\infty.$\qed

\textbf{Proof of Theorem \ref{pro3}}: For each $i\in \mathcal{A}=\{i\in\{1,\ldots, p\}: \mu_{xi}-\mu_{yi}\neq 0\}$, we can write $R^{0}_{i}$ as
\bea \label{a15}
R^{0}_{i}&=& \left|\frac{\bar{X}_i-\bar{Y}_i-\mu_{xi}-\mu_{yi}}{\sqrt{\hat{\sigma}^2_{xi}/n+\hat{\sigma}^2_{yi}/m}}+\frac{\mu_{xi}-\mu_{yi}}{\sqrt{\hat{\sigma}^2_{xi}/n+\hat{\sigma}^2_{yi}/m}} \right|\nonumber\\
&=& \left|S_{i}+\frac{\mu_{xi}-\mu_{yi}}{\sqrt{\sigma^2_{xi}/n+\sigma^2_{yi}/m}}+O_{p}(N^{-1/2}) \right|,\nonumber\\
\eea
where $S_{i}=\frac{\bar{X}_i-\bar{Y}_i-\mu_{xi}-\mu_{yi}}{\sqrt{\hat{\sigma}^2_{xi}/n+\hat{\sigma}^2_{yi}/m}}$ and (\ref{a15}) is obtained using the facts that $\hat{\sigma}^{2}_{xi}=\sigma^{2}_{xi}+O_{p}(N^{-1/2}),$ $\hat{\sigma}^{2}_{yi}=\sigma^{2}_{yi}+O_{p}(N^{-1/2}),$ and under assumption (A3), which implies that $\lambda_1=n/N=O(1)$, and $\lambda_2=m/N=O(1)$. For notational simplicity, denote $\frac{\mu_{xi}-\mu_{yi}}{\sqrt{\sigma^2_{xi}/n+\sigma^2_{yi}/m}}=\psi_{i}$. Using (\ref{a15}), we shall now establish the expression for P$\{R^{0}_i\leq x\}$ for any $x>0$. For each $i\in \mathcal{A}$, $\tilde{J}(R^{0}_{i})$ can be expressed as

\[
\tilde{J}_{i}(R^{0}_{i})=2\Phi(|S_{i}+\psi_{i}|)-1+o_{p}(1),  
\]
 since $S_i=O_{p}(1)$ , and $\hat{q}_{i}(|S_{i}+\psi_{i}|)\phi(|S_{i}+\psi_{i}|)=O_{p}(1)$, so $N^{-1}\hat{q}_{i}(|S_{i}+\psi_{i}|)\phi(|S_{i}+\psi_{i}|) \to 0$  as $N\to\infty$. If $i \notin \mathcal{A}$, then $\psi_{i}=0$ and $\tilde{J}_{i}(R^{0}_{i})$ can be expressed as
\[
\tilde{J}_{i}(R^{0}_{i})=2\Phi(|S_{i}+\psi_{i}|)-1+o_{p}(1)  
\]

as $N\to\infty.$ Now, $\max\limits_{1\leq i\leq p}\biggl\{2\Phi(|S_{i}+\psi_{i}|)-1+o_{p}(1)\biggr\}=2\Phi(\max\limits_{1\leq i\leq p}|S_{i}+\psi_{i}|)-1+o_{p}(1)$. Hence from (\ref{eq4}), we have   
\[2\Phi(\max_{1\leq i\leq p}|S_{i}+\psi_{i}|)-1+o_{p}(1)\approx 2\Phi(\max_{1\leq i\leq p}|\psi_{i}|)-1\] as $N\to\infty$, since $S_i=O_{p}(1)$. Using the Mills ratio , the alternative hypothesis $T^{\prime\prime}=[\max_{1\leq i\leq p}\Phi^{-1}\{\tilde{J}_{i}(R^{0}_i)\}]^2-2\text{log}(p)$+\text{log}(\text{log}$(p))$ can be approximated under  as follows 
\begin{multline}\label{ap16}
T^{\prime\prime}\approx\{\Phi^{-1}[2\Phi(\max_{1\leq i\leq p}|\psi_{i}|)-1]\}^2-\text{log}(p^2/\text{log}(p))\\
\approx \{\Phi^{-1}(2\biggl[1-\frac{\phi(\max_{1\leq i\leq p}|\psi_{i}|)}{\max_{1\leq i\leq p}|\psi_{i}|}\biggr]-1)\}^2-\text{log}(p^2/\text{log}(p))\\ 
=\{\Phi^{-1}\biggl(1-2\frac{\phi(\max_{1\leq i\leq p}|\psi_{i}|)}{\max_{1\leq i\leq p}|\psi_{i}|}\biggr)\}^2-\text{log}(p^2/\text{log}(p))\\
\approx 2\text{log}\biggl(\frac{\max_{1\leq i\leq p}|\psi_{i}|}{2\phi(\max_{1\leq i\leq p}|\psi_{i}|)}\biggr)-\text{log}(p^2/\text{log}(p)),
\end{multline}
last line is obtained using the fact that $\Phi^{-1}(1-1/t)=\sqrt{2\text{log}(t)}$, pp. 109 of \cite{26}; $\phi(x)$ the pdf of a standard normal  distribution. After simplifying (\ref{ap16}), we have the following expression
\bea \label{ap17}
T^{\prime\prime} \approx 2\text{log}\biggl(\max_{1\leq i\leq p}|\psi_{i}|\biggr)+\biggl(\max_{1\leq i\leq p}|\psi_{i}|\biggr)^2-\text{log}(p^2/\text{log}(p))-\text{log}(2/\pi),
\eea
as $N\to\infty.$

Again, (\ref{eq4}) implies that $\biggl(\max_{1\leq i\leq p}|\psi_{i}|\biggr)^2$ diverges faster than $\text{log}(p^2/\text{log}(p))$ as $N, p\to\infty$ so (\ref{ap17}) indicates that P$\{T^{\prime\prime}>q_{1-\alpha}\} \to 1$ as $N,p\to\infty.$
Hence we have Theorem \ref{pro3}. \qed

\newpage
\bibliographystyle{natbib}

\begin{thebibliography}{9}

\bibitem{1}
\textsc{McLachlan GJ , Do K-A , Ambriose C.} (2004): \textit{ Analyzing Microarray Gene Expression Data}. Wiley Sons Inc., Hoboken, NewJersey.

\bibitem{2}
\textsc{Eisen MB, Brown PO.} (1999): DNA arrays for analysis of gene expression. \textit{Methods in Enzymology,} 303, 179--204.


\bibitem{3}
\textsc{Supekar K, Menon V, Rubin D, Musen M, Greicius MD} (2008): Network analysis of intrinsic functional brain connectivity in Alzheimer's disease. \textit{PLoS Computional Biology}, 4,1--11.

\bibitem{4}
\textsc{Cai TT,  Liu W, Xia Y.  } (2014): Two-sample  test  of  high  dimensional  means  under dependence. \textit{Journal of the Royal Statistical Society, Series B}, 76, 349--372.

\bibitem{5}
\textsc{Hotelling H.} (1931):  The generalization of student’s ratio. \textit{The Annals of Mathematical Statistics,}, 2, 360--378.


\bibitem{6}
\textsc{Bai ZD, Saranadasa H} (1996):  Effect of high dimension:  By an example of a two sample problem. \textit{Statistica Sinica}, 6 ,311--329.

\bibitem{7}
\textsc{Chen SX, Qin Y.} (2010):  A two-sample test for high-dimensional data with applications to gene-set testing. textit{The Annals of Statistics}, 38, 808--835.

\bibitem{8}
\textsc{Srivastava MS,  DU M.} (2008):  A test for the mean vector with fewer observations than the dimension. \textit{Journal of Multivariate Analysis}, 99, 386--402.

\bibitem{9}
\textsc{Srivastava MS, Katayama S, Kano Y} (2013):  A two sample test in high dimensional data dimension. \textit{Journal of Multivariate Analysis}, 114, 349--358.


\bibitem{10}
\textsc{ Park J, Ayyala DN.} (2013): A test for the mean vector in large dimension and small samples. \textit{Journal of Statistical Planning and Inference}, 143, 929--943.

\bibitem{11}
\textsc{Gregory KB, Carroll RJ, Baladandayuthapani V, Lahiri, SN .} (2015): A Two-Sample Test for Equality of Means in High Dimension. \textit{Journal of the American Statistical Association}, 110, 837--849.









\bibitem{12}
\textsc{Beran R } (1987): Prepivoting to reduce level error of confidence sets, \textit{Biometrika}. 74, 457--468.

\bibitem{13}
\textsc{Beran R } (1988): Prepivoting Test Statistics: A Bootstrap View of Asymptotic Refinements. \textit{Journal of American Statistical Association}, 83, 687--697.


\bibitem{14}
\textsc{Romano J, Wolf M.} (2010): Balanced control of generalized error rates. \textit{The Annals of Statistics}, 38, 598--633.

\bibitem{15}
\textsc{Lunadron N.} (2015): Prepivoting composite score statistics by
weighted bootstrap iteration. \textit{The Canadian Journal of Statistics}, 43, 18--41.

\bibitem{16}
\textsc{Edgeworth, FY } (1905): The Law of error. \textit{Proceeding of the Cambridge Philollogical Society}, 20, 36--65.


\bibitem{17}
\textsc{Cram\'{e}r's, H} (1928): On the Composition. \textit{Skandinavisk Aktuarietidskrift}, 11, 13--74, 141--180.

\bibitem{18}
\textit {Bhattacharya R N, Ghosh, J K.} (1978): On the validity of the formal Edgeworth expansion. \textit{The Annals of Statistics}, 6,435--451.


\bibitem{19}
\textsc{Hall, P.} (1992): 
\textit{The Bootstrap and Edgeworth Expansion}. New York: Springer.


\bibitem{20}
\textsc{Song, PX} (2000): Multivariate Dispersion Models Generated
From Gaussian Copula.  \textit{Scandinavian Journal of Statistics}, 27, 305--320.

\bibitem{21}
\textsc{Hall P, Jing B-Y, Lahiri SN} (1998):  On the sampling window method for long-range dependent data.  \textit{Statistica Sinica}, 8, 1189--1204.

\bibitem{22}
\textsc{Goeman JJ, B\"{u}hlmann P} (2007): Analyzing gene expression data in terms of gene sets: methodological issues.  \textit{Bioinformatics}, 23, 980--987.

\bibitem{23}
\textsc{Alon, U, Barkai, N, Notterman, D A, Gish, K, Ybarra, S, Mack, D, and Levine, J} (1999): Broad patterns of gene expression revealed by clustering analysis of tumor and normal colon tissues probed by oligonucleotide arrays.,  \textit{Proc. Natl. Ecad. Sci., USA}, 96, 6745--6750.


\bibitem{24}
\textsc{Pomeroy, S L, Tamayo, P, Gaasenbeek, M, Sturla, L M, Angelo, M, McLaughlin, M E, Kim, J Y, Goumnerova, L C,  Black, P M, Lau, C,  et al.} (2002):Prediction of central nervous system embryonal tumour outcome based on gene expression. \textit{Nature}, 415,  436--442.

\bibitem{25}
\textsc{Serfling, RJ.} (1980): \textit{Approximation Theorems of Mathematical Statistics}. New York: John Wiley and Sons. 

\bibitem{26}
\textsc{DasGupta, A.} (2008): \textit{Asymptotic Theory of Statistics and Probability}. New York: Springer. 

\end{thebibliography}

\end{document}